\documentclass[12pt,a4paper]{article}
\usepackage[top=1in,bottom=1in,left=1in,right=1in,letterpaper]{geometry}
\usepackage{amsmath,amsthm,amssymb}
\usepackage{hhline}
\usepackage{algorithm}
\usepackage[noend]{algorithmic}
\usepackage{multirow}
\usepackage{bm}
\usepackage{authblk}
\usepackage{tikz}
\usetikzlibrary{arrows,positioning,decorations.pathreplacing,shapes}

\newenvironment{claimproof}{\noindent\emph{Proof of claim.\ }}{\hspace*{\fill}$\Box$\medskip}
\newtheorem{theorem}{Theorem}
\newtheorem{lemma}{Lemma}
\newtheorem{claim}{Claim}
\newtheorem{proposition}{Proposition}

\makeatletter\renewcommand{\ALG@name}{Pro\-ce\-du\-re}

\title{Primal-dual and dual-fitting analysis of online scheduling algorithms for generalized flow-time problems}

\author[1]{Spyros Angelopoulos}
\author[2,3]{Giorgio Lucarelli}
\author[4]{Nguyen Kim Thang}

\affil[1]{CNRS and Sorbonne Universit\'{e}s, UPMC Univ Paris 06, Paris, France \texttt{spyros.angelopoulos@lip6.fr}}
\affil[2]{LIG, University Grenoble-Alpes, France}
\affil[3]{LCOMS, Universit\'e de Lorraine, Metz, France
\texttt{giorgio.lucarelli@univ-lorraine.fr}}
\affil[4]{IBISC, Univ Evry, Universit\'e Paris-Saclay, 91025, Evry, France \texttt{thang@ibisc.fr}}

\date{}

\begin{document}

\maketitle

\begin{abstract}
We study online scheduling problems on a single processor that can be viewed as extensions of the well-studied problem of minimizing total weighted flow time.
In particular, we provide a framework of analysis that is derived by duality properties,
does not rely on potential functions and is applicable to a variety of scheduling problems.
A key ingredient in our approach is bypassing the need for ``black-box'' rounding of fractional solutions, which yields improved competitive ratios.

We begin with an interpretation of Highest-Density-First (HDF) as a primal-dual algorithm, and a corresponding
proof that HDF is optimal for total fractional weighted flow time (and thus scalable for the integral objective).
Building upon the salient ideas of the proof, we show how to apply and extend this analysis to the more general problem of
minimizing $\sum_j w_j g(F_j)$, where $w_j$ is the job weight, $F_j$ is the flow time and $g$ is a non-decreasing cost function.
Among other results, we present improved competitive ratios for the setting in which $g$ is a concave function,
and the setting of same-density jobs but general cost functions.
We further apply our framework of analysis to online weighted completion time with general cost functions
as well as scheduling under polyhedral constraints.
\end{abstract}

\section{Introduction}
\label{sec:introduction}

We consider online scheduling problems in which a set of jobs $\mathcal{J}$ arriving over time must be executed on a single processor.
In particular, each job $j \in \mathcal{J}$ is characterized by its \emph{processing time} $p_j>0$ and its \emph{weight} $w_j>0$,
which become known after its \emph{release time} $r_j \geq 0$.
The \emph{density} of $j$ is $\delta_j=w_{j}/p_{j}$, whereas, given a schedule, its {\em completion time}, $C_j$,
is defined as the first time $t\geq r_j$ such that $p_j$ units of $j$ have been processed.
The {\em flow time} of $j$ is then defined as $F_j=C_j-r_j$, and represents the time elapsed after the release of job $j$ and up to its completion.
A natural optimization objective is to design schedules that minimize the {\em total weighted flow time},
namely the sum $\sum_{j \in {\mathcal J}} w_jF_j$ of all processed jobs.
A related objective is to minimize the weighted sum of completion times, as given by the expression $\sum_{j \in {\mathcal J}} w_jC_j$.
We assume that {\em preemption} of jobs is allowed,
i.e., the jobs can be interrupted and be resumed at a later point without any penalty.

Im {\em et al.}~\cite{ImMoseley12:Online-scheduling} studied a generalization of the
total weighted flow time problem, in which jobs may incur non-linear contributions to the objective.
More formally, they defined the {\em Generalized Flow Time Problem} (\text{GFP}) in which the
objective is to minimize the sum $\sum_{j \in{\cal J}} w_j g(F_j)$, where $g:\mathbb{R}^+\rightarrow
\mathbb{R}^+$ is a given non-decreasing cost function with $g(0)=0$.
This extension captures many interesting and natural variants of flow time with real-life applications.
Moreover, it is an appropriate formulation of the setting in which we aim to simultaneously optimize several objectives.
We define the {\em Generalized Completion Time Problem} (GCP) along the same lines, with the only
difference being the objective function, which equals to $\sum_{j \in J} w_j g(C_j)$.

A further generalization of the above problems, introduced in~\cite{BansalPruhs10:The-Geometry-of-Scheduling},
associates each job $j$ with a non-decreasing cost function $g_j: \mathbb{R^+} \rightarrow \mathbb{R^+}$ and $g_j(0)=0$;
in the {\em Job-Dependent Generalized Flow Time Problem} (JDGFP), the objective is to minimize the sum $\sum_{j \in {\cal J}} w_j g_j(F_j)$.
This problem formulation is quite powerful, and captures many natural scheduling objectives such as weighted tardiness.
No competitive algorithms are known in the online setting;
approximation algorithms are studied in~\cite{BansalPruhs10:The-Geometry-of-Scheduling,holn2014:unsplittable}.

Very recently, Im {\em et al.}~\cite{ImKulkarni14:Competitive-Algorithms} introduced and studied a general
scheduling problem called {\em Packing Scheduling Problem} (PSP). Here, at any time $t$, the scheduler
may assign rates $\{x_j(t) \}$ to each job $j \in {\cal J}$. In addition, we are given a matrix $B$ of non-negative entries.
The goal is to minimize the total weighted flow time
subject to packing constraints $\{ B \mathbf{x} \leq 1, \mathbf{x}\geq 0\}$.
This formulates applications in which each job $j$ is associated with a resource-demand vector
$\mathbf{b}_j=(b_{1j}, b_{2j}, \ldots ,b_{Mj})$ so that it requires an amount $b_{ij}$ of the $i$-th resource.

In this paper, we present a general framework based on LP-duality principles, for online scheduling with generalized flow time objectives.
Since no online algorithm even for total weighted flow time is constant competitive~\cite{BansalChan09:Weighted-flow},
we study the effect of {\em resource augmentation}, introduced by Kalyanasundaram and Pruhs~\cite{KalyanasundaramPruhs00:Speed-is-as-powerful}.
More precisely, given some optimization objective (e.g. total flow time),
an algorithm is said to be $\alpha$-speed $\beta$-competitive if it is $\beta$-competitive with respect to
an offline optimal scheduling algorithm of speed $1/\alpha$ (here $\alpha \geq 1$).
\bigskip

\noindent
{\bf Related work.} \
It is well-known that the algorithm Shortest Remaining Processing Time (SRPT) is optimal for online total (unweighted) flow time.
Becchetti {\em et al.}~\cite{BecchettiLeonardi06:Online-weighted} showed that the natural algorithm Highest-Density-First (HDF)
is $(1+\epsilon)$-speed $\frac{1+\epsilon}{\epsilon}$-competitive for total weighted flow time.
At each time, HDF processes the job of highest density.

Concerning the online GFP, Im {\em et al.}~\cite{ImMoseley12:Online-scheduling}
showed that HDF is $(2+\epsilon)$-speed $O(\frac{1}{\epsilon})$-competitive algorithm for general non decreasing functions $g$.
On the negative side, they showed that no {\em oblivious} algorithm is $O(1)$-competitive with speed augmentation
$2-\epsilon$, for any $\epsilon >0$ (an oblivious algorithm does not know the function $g$).
In the case in which $g$ is a twice-differentiable, concave function, they showed that the algorithm Weighted Late Arrival Processor Sharing (WLAPS)
is $(1+\epsilon)$-speed $O(\frac{1}{\epsilon^2})$-competitive.
For equal-density jobs and general cost functions~\cite{ImMoseley12:Online-scheduling} prove that
FIFO is $(1+\epsilon)$-speed $\frac{4}{\epsilon^2}$-competitive.
Fox {\em et al.}~\cite{FoxIm13:Online-Non-clairvoyant} studied the problem of convex cost functions in the \emph{non-clairvoyant} variant
providing a $(2+\epsilon)$-speed $O(\frac{1}{\epsilon})$-competitive algorithm;
in this variant, the scheduler learns the processing time of a job only when the job is completed
(see for example ~\cite{Edmonds00:Scheduling-in-the-dark,EdmondsPruhs12:Scalably-scheduling,GuptaIm12:Scheduling-heterogeneous,GuptaKrishnaswamy10:Nonclairvoyantly-scheduling,KalyanasundaramPruhs00:Speed-is-as-powerful,MotwaniPhillips94:Non-Clairvoyant-Scheduling}).
Bansal and Pruhs~\cite{BansalPruhs04:Server-Scheduling} considered a special class of convex functions, namely the weighted $\ell_k$ norms of flow time,
with $1 < k < \infty$, and they showed that HDF is $(1+\epsilon)$-speed $O(\frac{1}{\epsilon^2})$-competitive.
Moreover, they showed how to transform this result in order to obtain an 1-speed $O(1)$-competitive algorithm
for the weighted $\ell_k$ norms of completion time.

Most of the above works rely to techniques based on amortized analysis (see also~\cite{ImMoseley11:A-tutorial-on-amortized} for a survey).
More recently, techniques based on LP duality have been applied in the context of online scheduling for generalized flow time problems.
Gupta {\em et al.}~\cite{GuptaKrishnaswamy12:Online-Primal-Dual} gave a primal-dual algorithm for a class of non-linear load balancing problems.
Devanur and Huang~\cite{DevanurHuang14:Primal-Dual} used a duality approach
for the problem of minimizing the sum of energy and weighted flow time on unrelated machines.
Of particular relevance to our paper is the work of Antoniadis {\em et al.}~\cite{antoniadis:STACS2014},
which gives an optimal {\em offline} energy and fractional weighted flow trade-off schedule for a speed-scalable processor with discrete speeds,
and uses an approach based on primal-dual properties (similar geometric interpretations arise in the context of our work, in the online setting).
Anand {\em et al.}~\cite{AnandGarg12:Resource-augmentation} were the first to propose
an approach to online scheduling by linear/convex programming and dual fitting.
Nguyen~\cite{Thang13:Lagrangian-Duality} presented a framework based on Lagrangian duality
for online scheduling problems beyond linear and convex programming.
The framework is applied to several problems related to flow time such as energy plus flow time, $\ell_{p}$-norm of flow time.
Im {\em et al.}~\cite{ImKulkarni14:Competitive-Algorithms} applied dual fitting in the context of PSP.
For the weighted flow time objective, they gave a non-clairvoyant algorithm that is $O(\log n)$-speed $O(\log n)$-competitive,
where $n$ denotes the number of jobs.
They also showed that for any constant $\epsilon>0$, any $O(n^{1-\epsilon})$-competitive algorithm
requires speed augmentation compared to the offline optimum.

We note that a common approach in obtaining a competitive, resource-augmented scheduling algorithm for flow time and related problems is by first deriving an algorithm that is competitive for the {\em fractional}
objective~\cite{BecchettiLeonardi06:Online-weighted,ImMoseley11:A-tutorial-on-amortized}.
An informal interpretation of the fractional objective is that a job contributes to the objective proportionally to the amount of its remaining work (see Section~\ref{sec:preliminaries} for a formal definition).
It is known that any $\alpha$-speed $\beta$-competitive algorithm for fractional GFP can be converted, in ``black-box'' fashion, to a $(1+\epsilon)\alpha$-speed $\frac{1+\epsilon}{\epsilon}\beta$-competitive algorithm
for (integral) GFP, for $0<\epsilon \leq 1$~\cite{FoxIm13:Online-Non-clairvoyant}.
Fractional objectives are often considered as interesting problems in their own (as in~\cite{antoniadis:STACS2014}).
\bigskip

\noindent
{\bf Contribution.} \
We present a framework for the design and analysis of algorithms for generalized flow time problems
that is based on primal-dual and dual-fitting techniques.
Our proofs are based on intuitive geometric interpretations of the primal/dual objectives; in particular, we do not rely on potential functions.
An interesting feature in our primal-dual approach, that differs from previous ones, is that
when a new job arrives, we may update the dual variables for jobs that already have been scheduled
without affecting the past portion (primal solution) of the schedule.
Another important ingredient of our analysis consists in relating, in a direct manner,
the primal integral and fractional dual objectives, without passing through the fractional primal.
This allows us to bypass the ``black-box'' transformation of fractional to integral solutions~\cite{BecchettiLeonardi06:Online-weighted,ImMoseley12:Online-scheduling},
which has been the canonical approach up to now.
As a result, we obtain an improvement to the competitive ratio by a factor of $O(\frac{1}{\epsilon})$, for $(1+\epsilon)$-speed.

\begin{table}
\hspace*{-1cm}
{\footnotesize
\begin{tabular}{|c||c|c|c|c|}
\hline
\multirow{2}{*}{$g(\cdot)$} & \multicolumn{2}{c|}{$\sum w_{j} g(F_{j})$}
 & \multicolumn{1}{c|}{\multirow{2}{*}{$\sum w_{j} g(C_{j})$}}
 & \multicolumn{1}{c|}{\multirow{2}{*}{$\sum w_{j} g_{j}(F_{j})$}} \\
\hhline{~--}
 & same density & arbitrary density & & \\
\hline
\hline
linear & & HDF, $(1+\epsilon,\frac{1+\epsilon}{\epsilon})$ \cite{BecchettiLeonardi06:Online-weighted} & & \\
\hline
\multirow{3}{*}{convex} & \multirow{3}{*}{\textbf{FIFO,} $\bm{(1+\epsilon,\frac{1+\epsilon}{\epsilon})}$}
 & WSETF, $(2+\epsilon, O(\frac{1}{\epsilon}))$ \cite{FoxIm13:Online-Non-clairvoyant} & & \\
 & & oblivious, $(2-\epsilon, \Omega(1))$ \cite{FoxIm13:Online-Non-clairvoyant} & & \\
 & & non-clairvoyant, $(\sqrt{2}-\epsilon, \Omega(1))$ \cite{FoxIm13:Online-Non-clairvoyant} & & \\
\hline
concave & & \multirow{2}{*}{WLAPS, $(1+\epsilon,O(\frac{1}{\epsilon^{2}}))$ \cite{ImMoseley12:Online-scheduling}} & & \multirow{2}{*}{$\bm{(1+\epsilon,\frac{4(1+\epsilon)^2}{\epsilon^{2}})}$} \\
differentiable & & & & \\
\hline
concave & \textbf{LIFO,} $\bm{(1+\epsilon,\frac{1+\epsilon}{\epsilon})}$ & \textbf{HDF,} $\bm{(1+\epsilon,\frac{1+\epsilon}{\epsilon})}$ & & \\
\hline
\multirow{2}{*}{general} & FIFO, $(1+\epsilon, \frac{4}{\epsilon^{2}})$ \cite{ImMoseley12:Online-scheduling}
 & HDF, $(2+\epsilon, O(\frac{1}{\epsilon}))$ \cite{ImMoseley12:Online-scheduling}
 & \multirow{2}{*}{\textbf{HDF,} $\bm{(1+\epsilon,\frac{1+\epsilon}{\epsilon})}$} & \\
 & \textbf{FIFO}, $\bm{(1+\epsilon,\frac{1+\epsilon}{\epsilon} })$ & $(7/6-\epsilon, \Omega(1))$ \cite{ImMoseley12:Online-scheduling} & & \\
\hline
\end{tabular}} \\[5pt]
  \caption{Summary of results for generalized flow time and completion time problems on a single machine. The
  $(\alpha,\beta)$ notation describes an algorithm which is $\alpha$-speed $\beta$-competitive.
  Our contribution (excluding our result on the PSP problem) is shown in bold.}
  \label{table:results}
\end{table}

In Section~\ref{section:flow} we begin with an interpretation of HDF as a primal-dual algorithm for total weighted flow time.
Our analysis, albeit significantly more complicated than the known combinatorial one~\cite{BecchettiLeonardi06:Online-weighted},
yields insights about more complex problems. One may draw parallels with the successful application of dual fitting in
approximation algorithms; for instance, the dual-fitting analysis of set cover that has led to improved algorithms for
metric facility location~\cite{Jain:Factor-revealing:2003}.
We also note that our approach differs from~\cite{DevanurHuang14:Primal-Dual}
(in which the objective is to minimize the sum of energy and weighted flow time),
even though the two settings are seemingly similar.
More precisely, the relaxation considered in~\cite{DevanurHuang14:Primal-Dual} consists only of covering constraints,
whereas for minimizing weighted flow~time, one has to consider both covering and packing constraints in the primal LP.

In Sections~\ref{section:framework-opt} and~\ref{sec:competitive} we expand the salient ideas behind the above analysis of HDF
and derive a framework which is applicable to more complicated objectives.
More precisely, we show that HDF is $(1+\epsilon)$-speed $\frac{1+\epsilon}{\epsilon}$-competitive for GFP with concave functions,
improving the $(1+\epsilon)$-speed $O(\frac{1}{\epsilon^2})$-competitive analysis of WLAPS~\cite{ImMoseley12:Online-scheduling}, and removing the assumption that $g$ is twice-differentiable.
For GFP with general cost functions and jobs of the same density, we show that FIFO is $(1+\epsilon)$-speed $\frac{1+\epsilon}{\epsilon}$-competitive,
which improves again the analysis in~\cite{ImMoseley12:Online-scheduling} by a factor of $O(\frac{1}{\epsilon})$ in the competitive ratio.
For the special case of GFP with equal-density jobs and convex (resp. concave) cost functions we show that FIFO (resp. LIFO) are fractionally optimal,
and $(1+\epsilon)$-speed $\frac{1+\epsilon}{\epsilon}$-competitive for the integral objective.

In addition, we apply our framework to the following problems:
i) online GCP: here, we show that HDF is optimal for the fractional objective,
and $(1+\epsilon)$-speed $\frac{1+\epsilon}{\epsilon}$-competitive for the integral one;
and ii) online PSP assuming a matrix $B$ of strictly positive elements:
here, we derive an adaptation of HDF which we prove is 1-competitive and which requires resource augmentation
$\max_j \frac{B_j}{b_j}$, with $B_j =\max_i b_{ij}$ and $b_j=\min_i b_{ij}$.

Last, in Section~\ref{sec:jdgfp} we extend ideas of~\cite{ImKulkarni14:SELFISHMIGRATE:-A-Scalable}, using, in addition, the Lagrangian
relaxation of a non-convex formulation for the online JDGFT problem.
We thus obtain a non-oblivious $(1+\epsilon)$-speed $\frac{4(1+\epsilon)^2}{\epsilon^2}$-competitive algorithm,
assuming each function $g_j$ is concave and differentiable.
Note that this result does not rely on our framework.

Table~\ref{table:results} summarizes the results of this paper in comparison to previous work.
\bigskip

\noindent
{\bf Notation.} \
Let $z$ be a job that is released at time $\tau$.
For a given scheduling algorithm, we denote by $P_\tau$ the set of pending jobs at time $\tau$
(i.e., jobs released up to and including $\tau$ but not yet completed),
and by $C_{\max}^\tau$ the last completion time among jobs in $P_\tau$, assuming no jobs are released after $\tau$.
We also define $R_\tau$ as the set of all jobs released up to and included $\tau$
(which may or may not have been completed at $\tau$)
and ${\cal J}_\tau$ as the set of all jobs that have been completed up to time $\tau$.

\section{Linear programming relaxation}
\label{sec:preliminaries}

In order to give a linear programming relaxation of GFP, we pass through the corresponding fractional variant of the objective.
Formally, let $q_{j}(t)$ be the remaining processing time of job $j$ at time $t$ (in a schedule).
The {\em fractional remaining weight} of $j$ at time $t$ is defined as $w_{j} q_{j}(t)/p_{j}$.
The fractional objective of GFP is now defined as $\sum_{j} \int_{r_{j}}^{\infty} w_{j} \frac{q_{j}(t)}{p_{j}} g'(t-r_{j})$.
Note that the fractional objective is a lower bound to the integral one, since 
\begin{equation*}
\sum_{j} \int_{r_{j}}^{\infty} w_{j} \frac{q_{j}(t)}{p_{j}} g'(t-r_{j})
 < \sum_{j} \int_{r_{j}}^{C_j} w_{j} g'(t-r_{j})
 < \sum_{j} w_{j} g(C_j-r_{j})
\end{equation*}
where the last inequality is due to the assumption that $g$ is non-decreasing.
An advantage of fractional GFP is that it admits a linear-programming formulation
(in fact, the same holds even for the stronger problem JDGFP (see Appendix~\ref{app:lp.formulation})).
Let $x_j(t) \in [0,1]$ be a variable that indicates the execution rate of $j \in \mathcal{J}$ at time $t$.
The primal and dual LPs are:

\hspace*{-1.5cm}
{\small
\begin{minipage}[t]{0.43\textwidth}
\begin{align}
\min \sum_{j \in \mathcal{J}} \delta_j \int_{r_j}^{\infty} &g(t-r_j) x_j(t) dt \tag{$P$} \\
\int_{r_j}^{\infty} x_j(t) dt &\geq p_j \quad \forall j \in \mathcal{J} \label{LP1}\\
\sum_{j \in \mathcal{J}} x_j(t) &\leq 1 \quad \forall t \geq 0 \label{LP2}\\
x_j(t) &\geq 0 \quad \forall j \in \mathcal{J}, t \geq 0 \notag
\end{align}
\end{minipage}
$\quad \quad$
\begin{minipage}[t]{0.55\textwidth}
\begin{align}
\max \sum_{j \in \mathcal{J}} \lambda_j & p_j - \int_0^{\infty} \gamma(t) dt \tag{$D$} \\
\lambda_j - \gamma(t) &\leq \delta_{j} g(t-r_j) & & \forall j \in \mathcal{J}, t \geq r_j \label{DLP1}\\
\lambda_j, \gamma(t) &\geq 0 & & \forall j \in \mathcal{J}, \forall t \geq 0 \notag
\end{align}
\end{minipage}
}
\bigskip

In this paper, we avoid the use of the ``black-box'' transformation (see for example \cite{BecchettiLeonardi06:Online-weighted,ImMoseley11:A-tutorial-on-amortized}) from fractional to integral GFP.
However, we also consider the fractional objective of GFP as a lower bound for the integral one.
Specifically, we will prove the performance of an algorithm by comparing its integral objective to that of a feasible dual solution $(D)$.
Note that by weak duality the latter is upper-bounded by the optimal solution of $(P)$, which equals to the optimal solution for fractional~GFP and hence it is a lower bound of the optimum solution for integral~GFP.

Moreover, we will analyze algorithms that are $\alpha$-speed $\beta$-competitive.
In other words, we compare the performance of our algorithm to an offline optimum with speed $1/\alpha$ ($\alpha, \beta>1$).
In turn, the cost of this offline optimum is the objective of a variant of $(P)$ in which constraints~(\ref{LP2}) are replaced by constraints $\sum_{j \in \mathcal{J}} x_j(t) \leq 1/\alpha$ for all $t \geq 0$.
The corresponding dual is the same as $(D)$, with the only difference that the objective is equal to
$\sum_{j \in \mathcal{J}} \lambda_j p_j - \frac{1}{\alpha}\int_0^{\infty} \gamma(t) dt$.
We denote these modified primal and dual LP's by $(P_\alpha)$ and $(D_\alpha)$, respectively.
In order to prove that the algorithm is $\alpha$-speed $\beta$-competitive, it will then be sufficient
to show that there is a feasible dual solution to $(D_\alpha)$ for which the
algorithm's cost is at most $\beta$ times the objective of the solution.

\section{A primal-dual interpretation of HDF for $\sum_j w_jF_j$}
\label{section:flow}

In this section we give an alternative statement of HDF as a primal-dual algorithm for the total weighted flow time problem.
We consider the primal and dual programs defined in the previous section where the function $g$ is considered to be the identity function, i.e., $g(x)=x$.
We begin with an intuitive understanding of the complementary slackness (CS) conditions.
In particular, the primal CS condition states that for a given job $j$ and time $t$,
if $x_j(t)>0$, i.e., if the algorithm were to execute job $j$ at time $t$, then it should be that $\gamma(t)=\lambda_j-\delta_j(t-r_j)$.
We would like then the dual variable $\gamma(t)$ to be such that we obtain some information about which job to schedule at time $t$.
To this end, for any job $j \in {\cal J}$, we define the line $\gamma_j(t)=\lambda_j-\delta_{j}(t-r_j)$, with domain $[r_j,\infty)$.
The slope of this line is equal to the negative density of the job, i.e, $-\delta_j$.
Our algorithm will always choose $\gamma(t)$ to be equal to $\max\{0,\max_{j \in \mathcal{J}: r_j \leq t}\{\gamma_j(t)\}\}$ for every $t \geq 0$.
We say that at time $t$ the line $\gamma_j$ (or the job $j$) is {\em dominant} if $\gamma_j(t)=\gamma(t)$;
informally, $\gamma_j$ is above $\gamma_{j'}$, for all $j' \neq j$.
We can thus restate the primal CS condition as a {\em dominance} condition: if a job $j$ is executed at time $t$, then $\gamma_j$ must be dominant at $t$.

We will consider a class of scheduling algorithms, denoted by ${\cal A}$, that comply to the following rules:
i) the processor is never idle if there are pending jobs; and
ii) if at time $\tau$ a new job $z$ is released, the algorithm will first decide an ordering on the set $P_\tau$ of all pending jobs at time $\tau$.
Then for every $t \geq \tau$, it schedules all jobs in $P_\tau$ according to the above ordering, unless a new job arrives after $\tau$.

We now proceed to give a primal-dual algorithm in the class ${\cal A}$ (which will turn out to be identical to HDF).
The algorithm will use the dominance condition so as to decide how to update the dual variables $\lambda_j$,
and, on the primal side, which job to execute at each time.
Note that once we define the $\lambda_j$'s, the lines $\gamma_j$'s as well as $\gamma(t)$ are well-defined, as we emphasized earlier.
In our scheme we change the primal and dual variables only upon arrival of a new job, say at time $\tau$.
We also modify the dual variables for jobs in ${\cal J}_\tau$, i.e., jobs that have already completed in the past (before time $\tau$)
without however affecting the primal variables of the past, so as to comply with the online nature of the problem.

\begin{figure}[h!]
\begin{center}
\begin{minipage}[t]{0.34\textwidth}
\begin{tikzpicture}[scale=1]
\draw[thick,->] (0,0) -- (4,0) node(xline)[below,font=\scriptsize] {$t$};
\draw[thick,->] (0,0) -- (0,3.5) node(yline)[left,font=\scriptsize] {$\gamma(t)$};
\draw[blue,dashed] (0,2.5) -- (1.8,0);
\draw[blue] (0,2.5) -- (1.4,0.55);
\draw[red,dashed] (0.5,1.1) -- (2.3,0);
\draw[red] (1.4,0.55) -- (2.3,0);
\draw[dotted] (0.0,1.1) -- (0.5,1.1); 
\node[font=\tiny] at (-0.2,2.5) {$\lambda_1$};
\node[font=\tiny] at (-0.2,1.1) {$\lambda_2$};
\draw[dotted] (1.4,0.0) -- (1.4,0.55);
\draw[dotted] (0.5,0.0) -- (0.5,1.1);
\node[font=\tiny] at (0,-0.2) {$r_1$};
\node[font=\tiny] at (0.5,-0.2) {$r_2$};
\node[font=\tiny] at (1.4,-0.2) {$C_1$};
\node[font=\tiny] at (2.3,-0.2) {$C_2$};
\end{tikzpicture}
\centerline{(a)}
\end{minipage}
\begin{minipage}[t]{0.29\textwidth}
\begin{tikzpicture}[scale=1]
\fill[blue!10] (0,3.14) -- (1.4,3.14) -- (1.4,1.2) -- (0,3.14);
\fill[red!10] (2,1.5) -- (2.9,1.5) -- (2.9,0) -- (2,0.575) --(2,1.5);
\fill[green!15] (1.4,1.72) -- (2,1.72) -- (2,0.575) -- (1.4,1.2) -- (1.4,1.72);
\draw[thick,->] (0,0) -- (4,0) node(xline)[below,font=\scriptsize] {$t$};
\draw[thick,->] (0,0) -- (0,3.5) node(yline)[left,font=\scriptsize] {$\gamma(t)$};
\draw[blue,dashed] (0,3.14) -- (2.26,0.0); 
\draw[blue] (0,3.14) -- (1.4,1.2); 
\draw[red,dashed] (0.5,1.5) -- (2.9,0.0); 
\draw[red] (2.0,0.575) -- (2.9,0.0); 
\draw[green!50!black,dashed] (0.9,1.72) -- (2.53,0.0); 
\draw[green!50!black] (1.4,1.2) -- (2,0.575); 
\draw[dotted] (0.0,3.14) -- (1.4,3.14); 
\draw[dotted] (0.0,1.5) -- (2.9,1.5); 
\draw[dotted] (0.0,1.72) -- (2.0,1.72); 
\node[font=\tiny] at (-0.2,3.14) {$\lambda_1$};
\node[font=\tiny] at (-0.2,1.5) {$\lambda_2$};
\node[font=\tiny] at (-0.2,1.72) {$\lambda_3$};
\draw[dotted] (1.4,0.0) -- (1.4,3.14); 
\draw[dotted] (2.9,0.0) -- (2.9,1.5); 
\draw[dotted] (2.0,0.0) -- (2.0,1.72); 
\draw[dotted] (0.5,0.0) -- (0.5,1.5); 
\draw[dotted] (0.9,0.0) -- (0.9,1.72); 
\node[font=\tiny] at (0,-0.2) {$r_1$};
\node[font=\tiny] at (0.5,-0.2) {$r_2$};
\node[font=\tiny] at (0.9,-0.2) {$r_3$};
\node[font=\tiny] at (1.4,-0.2) {$C_1$};
\node[font=\tiny] at (2.9,-0.2) {$C_2$};
\node[font=\tiny] at (2.0,-0.2) {$C_3$};
\end{tikzpicture}
\centerline{(b)}
\end{minipage}
\begin{minipage}[t]{0.34\textwidth}
\begin{tikzpicture}[scale=1]
\fill[blue!10] (0,0) -- (1.4,1.945) -- (1.4,0) -- (0,0);
\fill[red!10] (2,0) -- (2,0.91) -- (2.9,1.465) -- (2.9,0) -- (2,0);
\fill[green!15] (1.4,0) -- (1.4,0.46) -- (2,1.03) -- (2,0) -- (1.4,0);
\draw[thick,->] (0,0) -- (4,0) node(xline)[below,font=\scriptsize] {$t$};
\draw[thick,->] (0,0) -- (0,3.5) node(yline)[left,font=\scriptsize] {$\delta_j(t-r_j)$};
\draw[blue,dashed] (0,0) -- (2.52,3.5); 
\draw[blue] (0,0) -- (1.4,1.945); 
\draw[red,dashed] (0.5,0) -- (3.2,1.65); 
\draw[red] (2,0.91) -- (2.9,1.465); 
\draw[green!50!black,dashed] (0.9,0) -- (3.2,2.18); 
\draw[green!50!black] (1.4,0.46) -- (2,1.03); 
\draw[dotted] (1.4,0.0) -- (1.4,1.925); 
\draw[dotted] (2.9,0.0) -- (2.9,1.465); 
\draw[dotted] (2.0,0.0) -- (2.0,1.08); 
\node[font=\tiny] at (0,-0.2) {$r_1$};
\node[font=\tiny] at (0.5,-0.2) {$r_2$};
\node[font=\tiny] at (0.9,-0.2) {$r_3$};
\node[font=\tiny] at (1.4,-0.2) {$C_1$};
\node[font=\tiny] at (2.9,-0.2) {$C_2$};
\node[font=\tiny] at (2.0,-0.2) {$C_3$};
\end{tikzpicture}
\centerline{(c)}
\end{minipage}
\caption{Figure~(a) depicts the situation right before $\tau$:
the two lines $\gamma_1$, $\gamma_2$ correspond to two pending jobs prior to the release of $z$. In addition,
$\gamma(t)$ is the upper envelope of the two lines. Figure~(b) illustrates the situation
after the release of a third job $z$ at time $\tau=r_3$; the area of the shaded regions is the dual objective. In
Figure~(c), the area of the shaded regions is the primal fractional objective for the three jobs of Figure~(b).}
\label{fig:primaldual}
\end{center}
\end{figure}

By induction, suppose that the primal-dual algorithm $A \in {\cal A}$ satisfies the dominance condition up to time $\tau$, upon which a new job $z$ arrives.
Let $q_{j}$ be the remaining processing time of each job $j \in P_\tau$ at time $\tau$ and $|P_\tau|=k$.
Each $j \in P_\tau$ has a corresponding line $\gamma_j$, once $\lambda_j$ is defined.
To satisfy CS conditions, each line $\gamma_j$ must be defined such that to be dominant for a total period of time at least $q_j$, in $[\tau, \infty)$.
The crucial observation is that, if a line $\gamma_j$ is dominant at times $t_1, t_2$, it must also be dominant in the entire interval $[t_1,t_2]$.
This implies that for two jobs $j_1, j_2 \in P_\tau$, such that $j_1$ (resp. $j_2$) is dominant at time $t_1$ (resp. $t_2$), if $t_1 <t_2$
then the slope of $\gamma_{j_1}$ must be smaller than the slope of $\gamma_{j_2}$ (i.e., $-\delta_{j_1}\leq -\delta_{j_2}$).
We derive that $A$ must make the same decisions as HDF.
Consequently, the algorithm $A$ orders the jobs in $P_\tau$ in non-decreasing order of the slopes of the corresponding lines $\gamma_j$
(note that the slope of the lines is independent of the $\lambda_j$'s).
For every job $j \in P_\tau$, define $C_j= \tau + \sum_{j' \prec j} q_{j'}$, where the precedence is according to the above ordering of $A$.
These are the completion times of jobs in $P_\tau$ in $A$'s schedule, if no new jobs are released after time $\tau$;
so we set the primal variables $x_j(t)=1$ for all $t \in (C_{j-1}, C_j]$.
Procedure~\ref{algo:linear1} formalizes the choice of $\lambda_j$ for all $j \in P_\tau$;
intuitively, it ensures that if a job $j \in P_\tau$ is executed at
time $t>\tau$ then $\gamma_j$ is dominant at $t$ (see Figure~\ref{fig:primaldual} for an illustration).

\begin{algorithm}[h]
\begin{algorithmic}[1]
\STATE Consider the jobs in $P_\tau$ in increasing order of completion times if no new jobs are released after time $\tau$, i.e., $C_1<C_2<\ldots <C_k$
\STATE Choose $\lambda_{k}$ such that $\gamma_{k}(C_{k}) = 0$
\FOR{each pending job $j = k-1$ to $1$}
    \STATE Choose $\lambda_{j}$ such that $\gamma_{j}(C_{j}) = \gamma_{j+1}(C_{j})$
\ENDFOR
\end{algorithmic}
\caption{Assignment of dual variables $\lambda_j$ for all $j \in P_\tau$ at the arrival of a new job $z$.}
\label{algo:linear1}
\end{algorithm}

We first observe a monotonicity property of $\lambda_j$'s.
 
\begin{lemma} \label{lemma:monotonicitylinear}
By Procedure~\ref{algo:linear1}, the value of the dual variable $\lambda_j$, $j \in P_\tau$, can be only increased after the arrival of a new job $z$ at time $\tau$.
\end{lemma}
\begin{proof}
Let $\lambda'_{j}$ and $\lambda_{j}$ be the value of the dual variable of $j$ before and after the arrival of $z$, respectively.
We consider the following two cases.

\begin{description} 
\item[Job $j$ is delayed by $z$.]
Let $C_{j}$ be the new completion time of job $j$. So before the arrival of $z$, the completion of $j$ was $C_{j} - p_{z}$.   
Thus, by Procedure~\ref{algo:linear1}, we have that $\gamma_j(C_j)=\gamma'_j(C_j-p_z)$ where $\gamma'_{j}$ is the corresponding line of $j$ before the arrival of $z$.
Thus, $\lambda_{j} - \delta_{j}(C_{j} - r_{j}) = \lambda'_{j} - \delta_{j}(C_{j} - p_{z} - r_{j})$, that is $\lambda_{j} > \lambda'_{j}$.
\item[Job $j$ is not delayed by $z$.]
In this case, Procedure~\ref{algo:linear1} increases the dual variable of $j$
by $\gamma_{z}(C_{z} - p_{z} - r_{z}) - \gamma_{z}(C_{z} - r_{z}) > 0$.
Hence, $\lambda_{j} > \lambda'_{j}$.
\end{description}
The lemma follows.
\end{proof}

\begin{figure}[h!]
\begin{tikzpicture}[scale=1]
\draw[thick,->] (0,0) -- (8,0) node(xline)[below,font=\scriptsize] {$t$};
\draw[thick,->] (0,0) -- (0,4.2) node(yline)[left,font=\scriptsize] {$\gamma(t)$};
%
\node[font=\tiny] at (0,-0.2) {$r_1$};
\node[font=\tiny] at (0.8,-0.2) {$C_1$};
\node[font=\tiny] at (-0.2,3.9-0.489) {$\lambda_1$};
\draw[blue,dashed] (0,3.9-0.489) -- (1.399,0);
\draw[blue,thick] (0,3.9-0.489) -- (0.8,1.95-0.489);
\draw[dotted] (0.8,1.95-0.489) -- (0.8,0); 

\node[font=\tiny] at (0.3,-0.2) {$r_2$};
\node[font=\tiny] at (5.9,-0.2) {$C_2$};
\node[font=\tiny] at (-0.2,2.1-0.489) {$\lambda_2$};
\draw[red,dashed] (0.3,2.1-0.489) -- (5.9,0);
\draw[red,thick] (0.8,1.95-0.489) -- (1.8,1.66233-0.489);
\draw[red,thick] (3.9,0.57534) -- (5.9,0);
\draw[dotted] (0.0,2.1-0.489) -- (0.3,2.1-0.489); 
\draw[dotted] (0.3,0.0) -- (0.3,2.1-0.489); 

\node[font=\tiny] at (1.8,-0.2) {$r_3$};
\node[font=\tiny] at (3.1,-0.2) {$C_3$};
\node[font=\tiny] at (-0.2,4.6-0.84) {$\lambda_3$};
\draw[yellow!60!black,dashed] (1.8,4.6-0.84) -- (4.3-0.84*2.5/4.6,0);
\draw[yellow!60!black,thick] (1.8,4.6-0.84) -- (3.1,2.208-0.84);
\draw[dotted] (0.0,4.6-0.84) -- (1.8,4.6-0.84); 
\draw[dotted] (1.8,0.0) -- (1.8,4.6-0.84); 
\draw[dotted] (3.1,0.0) -- (3.1,2.208-0.84); 

\node[font=\tiny] at (2.4,-0.2) {$r_4$};
\node[font=\tiny] at (3.9,-0.2) {$C_4$};
\node[font=\tiny] at (-0.2,2.91054-0.84) {$\lambda_4$};
\draw[orange!80,dashed] (2.4,2.91054-0.84) -- (5.3-0.84*2.2/2.208,0);
\draw[orange!80,thick] (3.1,2.208-0.84) -- (3.9,1.4051-0.84);
\draw[dotted] (0.0,2.91054-0.84) -- (2.4,2.91054-0.84); 
\draw[dotted] (2.4,0.0) -- (2.4,2.91054-0.84); 
\draw[dotted] (3.9,0.0) -- (3.9,1.4051-0.84); 
\end{tikzpicture}
\\
\begin{tikzpicture}[scale=1]
\draw[thick,->] (0,0) -- (8,0) node(xline)[below,font=\scriptsize] {$t$};
\draw[thick,->] (0,0) -- (0,5) node(yline)[left,font=\scriptsize] {$\gamma(t)$};
%
\node[font=\tiny] at (0,-0.2) {$r_1$};
\node[font=\tiny] at (0.8,-0.2) {$C_1$};
\node[font=\tiny] at (-0.2,3.9) {$\lambda_1$};
\draw[blue,dashed] (0,3.9) -- (1.6,0);
\draw[blue,thick] (0,3.9) -- (0.8,1.95);
\draw[dotted] (0.8,1.95) -- (0.8,0); 

\node[font=\tiny] at (0.3,-0.2) {$r_2$};
\node[font=\tiny] at (7.6,-0.2) {$C_2$};
\node[font=\tiny] at (-0.2,2.1) {$\lambda_2$};
\draw[red,dashed] (0.3,2.1) -- (7.6,0);
\draw[red,thick] (0.8,1.95) -- (1.8,1.66233);
\draw[red,thick] (5.6,0.58142) -- (7.6,0);
\draw[dotted] (0.0,2.1) -- (0.3,2.1); 
\draw[dotted] (0.3,0.0) -- (0.3,2.1); 

\node[font=\tiny] at (1.8,-0.2) {$r_3$};
\node[font=\tiny] at (3.1,-0.2) {$C_3$};
\node[font=\tiny] at (-0.2,4.6) {$\lambda_3$};
\draw[yellow!60!black,dashed] (1.8,4.6) -- (4.3,0);
\draw[yellow!60!black,thick] (1.8,4.6) -- (3.1,2.208);
\draw[dotted] (0.0,4.6) -- (1.8,4.6); 
\draw[dotted] (1.8,0.0) -- (1.8,4.6); 
\draw[dotted] (3.1,0.0) -- (3.1,2.208); 

\node[font=\tiny] at (2.4,-0.2) {$r_4$};
\node[font=\tiny] at (3.9,-0.2) {$C_4$};
\node[font=\tiny] at (-0.2,2.91054) {$\lambda_4$};
\draw[orange!80,dashed] (2.4,2.91054) -- (5.3,0);
\draw[orange!80,thick] (3.1,2.208) -- (3.9,1.4051);
\draw[dotted] (0.0,2.91054) -- (2.4,2.91054); 
\draw[dotted] (2.4,0.0) -- (2.4,2.91054); 
\draw[dotted] (3.9,0.0) -- (3.9,1.4051); 

\node[font=\tiny] at (3.6,-0.2) {$r_5$};
\node[font=\tiny] at (5.6,-0.2) {$C_5$};
\node[font=\tiny] at (-0.2,1.55045) {$\lambda_5$};
\draw[green!50!black,dashed] (3.6,1.55045) -- (6.8,0);
\draw[green!50!black,thick] (3.9,1.4051) -- (5.6,0.58142);
\draw[dotted] (0.0,1.55045) -- (3.6,1.55045); 
\draw[dotted] (3.6,0.0) -- (3.6,1.55045); 
\draw[dotted] (5.6,0.0) -- (5.6,0.58142); 
\end{tikzpicture}
\caption{The figure depictes the lines corresponding to the dual variables before (above) and after (below) the arrival of job 5 at time $\tau=r_5$.
The set of completed jobs at time $\tau$ is $\mathcal{J}_\tau=\{1,3\}$.
Moreover, Procedure~\ref{algo:linear2} defines three disjoint sets: $\{1,2\}$ with representative 2, $\{5\}$ with representative 5 and $\{3,4\}$ with representative 4.}
\label{fig:sets}
\end{figure}

The following lemma shows that if no new jobs were to be released after time $\tau$, HDF would guarantee the dominance
condition for all times $t \geq \tau$.

\begin{lemma}[future dominance] \label{lemma:dominancelinear}
For $\lambda_{j}$'s as defined by Procedure~\ref{algo:linear1}, and $A \equiv HDF$,
if job $j \in P_\tau$ is executed at time $t \geq \tau$, then $\gamma_j$ is dominant at $t$,
assuming that no new jobs are released after time $\tau$.
\end{lemma}
\begin{proof}
Suppose that the jobs in $P_\tau$ are in increasing order of their completion times, i.e., $C_1 < C_2 < \ldots < C_k$.
In order to prove the lemma, we have to show that
\begin{eqnarray}
\gamma_j(t) \leq \gamma_{j+1}(t) & & \forall t \geq C_{j}
\label{eq:dominant-ineq-1} \\
\gamma_j(t) \geq \gamma_{j+1}(t) & & \forall t \in [\tau,C_{j}]
\label{eq:dominant-ineq-2}
\end{eqnarray}
Recall that by the choice of $\lambda_{j}$'s, we have that $\lambda_{j} - \delta_{j}(C_{j}-r_{j}) = \lambda_{j+1} - \delta_{j+1}(C_{j}-r_{j+1})$.
Therefore,
\begin{eqnarray}
\lambda_{j} - \delta_{j}(t-r_{j}) & = & \lambda_{j+1} - \delta_{j+1}(C_{j}-r_{j+1}) + \delta_{j}(C_{j}-r_{j}) - \delta_{j}(t-r_{j}) \notag \\
 & = & \lambda_{j+1} - \delta_{j+1}(t-r_{j+1}) + (\delta_j-\delta_{j+1}) (C_{j}-t) \label{eq:dominant-critical}
\end{eqnarray}
Since $A \equiv HDF$, we have that $\delta_j - \delta_{j+1} \geq 0$.
If $t \geq C_j$ then (\ref{eq:dominant-ineq-1}) follows, while if $t < C_j$ then (\ref{eq:dominant-ineq-2}) follows.
\end{proof}

By Lemma~\ref{lemma:monotonicitylinear}, Procedure~\ref{algo:linear1} modifies (increases) the $\lambda_j$ variables of all jobs pending at time $\tau$.
In turn, this action may violate the dominance condition {\em prior to $\tau$}.
We thus need a second procedure that will rectify the dominance condition for $t \leq \tau$.

We consider again the jobs in $P_\tau$ in increasing order of their completion times, i.e., $C_1 < C_2 < \ldots < C_k$, with $k=|P_\tau|$.
We partition $R_{\tau}$ into $k$ disjoint sets $S_1,S_2,\ldots,S_{k}$.
Each set $S_j$ is initialized with the job $j \in P_\tau$, which is called the \emph{representative} element of $S_{j}$
(we use the same index to denote the set and its representative job).
Informally, the set $S_{j}$ will be constructed in such a way that it will contain
all jobs $a \in \mathcal{J}_{\tau}$ whose corresponding variable $\lambda_a$ will be increased by the same amount in the procedure.
This amount is equal to the increase, say $\Delta_j$, of $\lambda_j$, due to Procedure~\ref{algo:linear1} for the representative job of~$S_j$.

Fig.~\ref{fig:sets} gives an example of the definition of the above sets.
We can also observe that the distance between $\lambda_1$ and $\lambda_2$, as well as, between $\lambda_3$ and $\lambda_4$ remains the same before and after the arrival of the job 5.
However, this is not true for jobs of different sets (see for example $\lambda_1$ and $\lambda_3$).
We then define Procedure~\ref{algo:linear2} that describes formally the increase in the dual variables for jobs in~$\mathcal{J}_{\tau}$.

\begin{algorithm}
\begin{algorithmic}[1]
\FOR {$j=1$ to $k$}
    \STATE Add $j$ in $S_j$
\ENDFOR
\FOR {each job $a \in \mathcal{J}_{\tau}$ in decreasing order of completion times (defined under the assumption that no new jobs are released after time $\tau$)}
    \STATE Let $b$ be the job such that $\gamma_{a}(C_{a})=\gamma_{b}(C_{a})$
    \STATE Let $S_j$ be the set that contains $b$
    \STATE Add $a$ in $S_j$
\ENDFOR
\FOR {each set $S_{j}$, $1 \leq j \leq k$,}
    \STATE Let $\Delta_j$ be the increase of $\lambda_j$, due to Procedure~\ref{algo:linear1}, for the representative of $S_j$
    \STATE Increase $\lambda_{a}$ by an amount of $\Delta_{j}$ for all $a \in S_{j} \setminus \{j\}$
\ENDFOR
\end{algorithmic}
\caption{Updating of dual variables $\lambda_j$ for all jobs $j\in {\cal J}_\tau$ at the arrival of a new job $z$.}
\label{algo:linear2}
\end{algorithm}

Geometrically, the update operation is a vertical translation of the line $\gamma(t)$ for $t < \tau$.
The following lemma ensures that the job $b$ selected in Line~4 of Procedure~\ref{algo:linear2} always exists.

\begin{lemma}
For each job $a \in \mathcal{J}_{\tau}$, there is always a job $b$ such that $\gamma_{a}(C_{a})=\gamma_{b}(C_{a})$ (except if $a$ is the last completed job of the schedule).
Moreover, Procedure~\ref{algo:linear2} adds $b$ to a set $S_j$ before treating $a$.
\end{lemma}
\begin{proof}
Let $c$ be the last job arrived before the completion of $a$, i.e., $r_c \leq C_a$ and there is no job $c'$ with $r_c < r_{c'} \leq C_a$.
Note that $a$ may coincide with $c$.
Since $a \in P_{r_c}$, there is a job $b$ such that $\gamma_{a}(C_{a})=\gamma_{b}(C_{a})$ as this is defined by Line~4 of Procedure~\ref{algo:linear1}, except if $a$ is the last completed job of the schedule (Line~2 of Procedure~\ref{algo:linear1}).
After this point only Procedure~\ref{algo:linear2} can raise this property.
However, Procedure~\ref{algo:linear2} always groups the jobs $a$ and $b$ in the same set, let $S_j$.
Since all lines corresponding to jobs in $S_j$ are vertically translated by the same quantity, the first part of the lemma follows.

If $b \in P_\tau$, then the second part of the lemma directly holds.
Otherwise, consider again the last job, $c$, arrived before the completion of $a$.
By the definition of $c$ and the execution of Procedure~\ref{algo:linear1} at time $r_c$, we know that $a$ completes before $b$ and hence it is treated before $a$ by Procedure~\ref{algo:linear2}.
\end{proof} 

The following lemma shows that, if a line $\gamma_j$ was dominant for a time $t < \tau$ prior to the arrival of the new job at time $\tau$, then it will remain dominant after the application of Procedures~\ref{algo:linear1} and~\ref{algo:linear2}.

\begin{lemma}[past dominance] \label{lemma:consistencylinear}
For $\lambda_{j}$'s as defined by both Procedure~\ref{algo:linear1} and Procedure~\ref{algo:linear2}, and $A \equiv HDF$,
if job $j \in {\cal J}_\tau \cup P_\tau$ is executed at time $t<\tau$, then $\gamma_j$ is dominant at $t$.
\end{lemma}
\begin{proof}
The proof is based on the following three claims.

\begin{claim}\label{claim:linear-change}
Let $j_{1}$ and $j_{2}$ be two jobs in $P_\tau$ such that $C_{j_1} < C_{j_2}$.
If $A \equiv HDF$, then $\Delta_{j_{1}} \geq \Delta_{j_{2}}$.
\end{claim}
\begin{claimproof}
Let $z$ be the job released at time $\tau$, with processing time $p_z$.
Consider the following three cases.
\begin{itemize}
\item[(i)] $C_z <C_{j_1}<C_{j_2}$.
Hence, the completion times of both $j_{1}$ and $j_{2}$ are delayed by $p_{z}$ by HDF.
In other words, before the arrival of $z$ the completion times of $j_{1}$ and $j_{2}$ were $C_{j_{1}}-p_{z}$
and $C_{j_{2}}-p_{z}$, respectively.
Moreover, the relative orderings of all jobs in $P_\tau \setminus \{z\}$ (as done by the algorithm)
is the same before and after the release of $z$.
Thus, by Procedure~\ref{algo:linear1}, we have that
$\lambda_{j_{1}} - \delta_{j_{1}}(C_{j_{1}} - r_{j_{1}}) = \lambda'_{j_{1}} - \delta_{j_{1}}(C_{j_{1}} - p_{z} - r_{j_{1}})$,
where $\lambda'_{j_{1}}$ is the value of the dual variable of $j_{1}$ before the arrival of $z$.
Hence, $\Delta_{j_{1}} = \delta_{j_{1}} p_{z}$.
Similarly, we obtain that $\Delta_{j_{2}} = \delta_{j_{2}} p_{z}$.
Therefore, $\Delta_{j_{1}} \geq \Delta_{j_{2}}$ as from the HDF algorithm we infer that
$\delta_{j_1} \geq \delta_{j_2}$.
\item[(ii)] $C_{j_1}<C_{j_2}<C_{z}$.
In this case, Procedure~\ref{algo:linear1} increases both $\lambda_{j_1}$ and $\lambda_{j_2}$
by $\gamma_{z}(C_{z} - p_{z} - r_{z}) - \gamma_{z}(C_{z} - r_{z})$.
Therefore, $\Delta_{j_{1}} = \Delta_{j_{2}} = \delta_{z} p_{z}$.
\item[(iii)] $C_{j_1}<C_z<C_{j_2}$.
As in case (ii), we have $\Delta_{j_{1}} = \delta_{z} p_{z}$.
As in case (i), we have $\Delta_{j_{2}} = \delta_{j_{2}} p_{z}$.
Therefore, $\Delta_{j_{1}} \geq \Delta_{j_{2}}$ as from the HDF algorithm we infer that
$\delta_{z} \geq \delta_{j_{2}}$.
\end{itemize}
The claim follows.
\end{claimproof}

We call a set $S_j$ \emph{critical} if at least one of the following hold:
$S_j$ contains at least one job in $\mathcal{J}_{\tau}$ (i.e., $S_j \cap \mathcal{J}_{\tau} \not= \emptyset$)
or its representative job has been partially executed before time $\tau$ (i.e., $q_j(\tau)<p_j$).
Let $\ell \leq k$ be the number of critical sets.
Let also $Q_\tau \subseteq P_\tau$ denote the set of representative jobs of critical subsets; hence $|Q_\tau|=\ell$.
Note that job $z$ which is released at time $\tau$ does not belong in $Q_\tau$.

\begin{claim}\label{claim:linear-order}
Let $j_{1}$ and $j_{2}$ be two jobs in $Q_\tau$ such that $C_{j_1} < C_{j_2}$.
Then $r_{j_{1}} \geq r_{j_{2}}$.
\end{claim}
\begin{claimproof}
By way of contradiction, suppose that there are two jobs $j_1$ and $j_2$ in $Q_{\tau}$
such that $C_{j_1} < C_{j_2}$ and $r_{j_{1}} < r_{j_{2}}$.
Since $C_{j_1}<C_{j_2}$, $j_1$ has higher density than $j_2$.
Hence, $j_2$ has not been scheduled before $\tau$ because $j_1$ is active during $[r_{j_2},\tau]$ and it has higher density.
Thus, from the definition of $Q_\tau$, there is a job $a \in S_{j_2} \cap \mathcal{J}_{\tau}$
for which $\gamma_a(C_{a})=\gamma_{j_2}(C_{a})$ and $C_a < \tau$.
Therefore, at time $C_a$ the job $j_2$ is the pending job with the highest density,
which is a contradiction as we assumed that $j_1$ is already released by time $C_a$ and that
it has higher density than $j_2$.
\end{claimproof}

Let $t_j$ be the first time in which a job in the critical set $S_j$ begins its execution for $1 \leq j \leq \ell$.
Let $i_{1}, i_2, \ldots, i_{\ell}$ be job indices such that $t_{i_{1}} < t_{i_2} < \ldots < t_{i_{\ell}} \leq \tau$.
The following lemma shows a structure property of the algorithm schedule.

\begin{claim}\label{claim:linear-continuous}
During interval $[t_{i_{j}}, t_{i_{j+1}})$, only jobs in $S_{j}$ are executed for $1 \leq j \leq \ell - 1$.
\end{claim}
\begin{claimproof}
Consider interval $[t_{i_{j}}, t_{i_{j+1}})$ and let $a$ be the job processed at time $t_{i_{j}}$.
If $a$ is still pending at time $\tau$ then by definition $a$ is indeed job $i_{j}$ and $S_{i_{j}}$ consists of a singleton job.
If $a$ is completed at time $C_{a} < \tau$ then by Procedure \ref{algo:linear2} there exists a job $b$ such that $\gamma_{b}(C_{a}) = \gamma_{a}(C_{a})$ and jobs $a$ and $b$ belong to the same representative set.
By repeating the same argument for job $b$ inductively, we infer that all jobs executed in interval $[t_{i_{j}}, t_{i_{j+1}})$ belong to the same representative set.
\end{claimproof}

We now continue with the proof of the lemma and we show that at each time $t < \tau$, the dominant job remains the same before and after the application of Procedures~\ref{algo:linear1} and~\ref{algo:linear2}.
By Procedure~\ref{algo:linear2}, the dual variables $\lambda_j$ of all jobs in the same critical set are all increased by the same amount.
Moreover, by Claim~\ref{claim:linear-continuous} the jobs in the same critical set are all executed consecutively.
Consider two jobs $j_1$ and $j_2$ in the same critical set $S_j$.
Assume, without loss of generality, that $\gamma'_{j_1}(t) \geq \gamma'_{j_1}(t)$ for time $t< \tau$, where $\gamma'_{j_1}$ and $\gamma'_{j_2}$ are the lines of $j_1, j_2$ prior to the arrival of $z$.
Then $\gamma_{j_1}(t)=\gamma'_{j_1}(t)+\Delta_j\geq\gamma'_{j_2}(t)+\Delta_j=\gamma_{j_2}(t)$.
Therefore, it suffices to consider only jobs that belong in different critical sets.

Consider the critical sets in decreasing order of completion times of their representatives, i.e., $C_1 > C_2 > \ldots > C_{\ell}$, and let $S_j$ and $S_{j+1}$ be any pair of consecutive critical sets.
Consider also the time $t_{i_{j+1}} < \tau$ as it is defined before Claim~\ref{claim:linear-continuous}.
Let $a \in S_j$ be the last job that is executed before $t_{i_{j+1}}$ and $b \in S_{j+1}$ be the job that is executed at $t_{i_{j+1}}$.
By Claim~\ref{claim:linear-change}, $\Delta_j \leq \Delta_{j+1}$ and hence $\lambda_a$ has been increased at most as much as $\lambda_b$ did (in Procedure~\ref{algo:linear2}).
Thus, $\gamma_b$ is dominant after $t_{i_{j+1}}$ and it remains to show that $\gamma_a$ is dominant just before $t_{i_{j+1}}$.
In order to prove this, it suffices to prove that the job $b$ is not yet released by time $t_{i_{j+1}}$, which means that $\gamma_b$ does not affect the line $\gamma_a$ prior to this time.
Indeed, if $b$ was released before $t_{i_{j+1}}$, then $C_a=t_{i_{j+1}}$ and $\Delta_a>\Delta_b$ since $A \equiv HDF$ and $a$ was dominant just before $t_{i_{j+1}}$ while $b$ was dominant after $t_{i_{j+1}}$ when considering the situation prior to the arrival of the new job $z$ at time $\tau$.
Hence, Procedure~\ref{algo:linear1} would have defined $\gamma_a(C_a)\not=\gamma_b(C_a)$, that is $a$ and $b$ would belong to the same representative set, which is a contradiction.
\end{proof}

Recall that $C_{\max}^\tau$ denotes the completion time of the last pending job in $P_\tau$
(with the usual assumption that no job arrives after time $\tau$).
The following lemma states that the dual variable $\gamma(t)$ has been defined in such a way that it is zero for all $t> C_{\max}^\tau$.
This will be required in order to establish that the primal and dual solutions have the same objective value.

\begin{lemma}[completion] \label{lemma:completionlinear}
For $\lambda_{j}$'s defined by Procedures~\ref{algo:linear1} and~\ref{algo:linear2},
we have that $\gamma(t)=0$ for every $t > C_{\max}^\tau$.
\end{lemma}
\begin{proof}
From Lemmas~\ref{lemma:dominancelinear} and~\ref{lemma:consistencylinear},
for every job $j \in \mathcal{J}_\tau \cup P_\tau$, it holds that $\gamma_{j}(C_{\max}^\tau) \leq \gamma(C_{\max}^\tau)$.
By construction we have that $\gamma(C_{\max}^\tau) = 0$, and hence $\gamma_{j}(C_{\max}^\tau) \leq 0$.
Moreover, $\gamma_j(t)$ is non-increasing function of $t$,
that is $\gamma_{j}(t) \leq \gamma_{j}(C_{\max}^\tau)$ for every $t \geq C_{\max}^\tau$.
Therefore, $\gamma(t)=\max\{0,\max_{j:r_j \leq t}\{\gamma_j(t)\}\} = 0$ for $t \geq C_{\max}^\tau$.
\end{proof}

The proof of the following theorem is based on Lemmas~\ref{lemma:dominancelinear},~\ref{lemma:consistencylinear} and~\ref{lemma:completionlinear}, and
it is a simplified case of the proof of Theorem~\ref{thm:framework1} which is given in the next section.
We note that the primal and dual objectives have intuitive geometric interpretations, as shown in Figure~\ref{fig:primaldual}(b),(c)).
For each job, its contribution to the dual objective is the area of a trapezoid that is exactly the same as the contribution of the job to the primal objective.

\begin{theorem} \label{thm:linear}
The primal-dual algorithm $A \equiv HDF$ is an optimal online algorithm for the total fractional weighted flow time
and a $(1+\epsilon)$-speed $\frac{1+\epsilon}{\epsilon}$-competitive algorithm for the total (integral) weighted flow time.
\end{theorem}

\section{A framework for primal-dual algorithms}
\label{section:framework-opt}

Building on the primal-dual analysis of HDF for total weighted flow time,
we can abstract the essential properties that we need to satisfy in order to obtain online algorithms for other similar problems.
For the problems we consider, the primal solution is generated by an online primal-dual algorithm $A \in\mathcal{A}$ which may not necessarily be HDF.
In addition, each job $j$ will now correspond to a {\em curve} $\gamma_j$ (for the total weighted flow time problem, $\gamma_j$ is a line),
and we will also have a dual variable $\gamma(t)$ that will be set equal to $\max\{0,\max_{j \in \mathcal{J}: r_j \leq t}\{\gamma_j(t)\}\}$ for every $t \geq 0$.
Finally, the crux is in maintaining dual variables $\lambda_j$, upon release of a new job $z$ at time $\tau$,
such that the following properties are satisfied:
\begin{description}
\item[$(\mathcal{P}1)$ \textbf{Future dominance}.] If the algorithm $A$ executes job $j$ at time $t \geq \tau$, then $\gamma_j$ is dominant at $t$.
\item[$(\mathcal{P}2)$ \textbf{Past dominance}.] If the algorithm $A$ executes job $j$ at time $t< \tau$, then $\gamma_j$ remains dominant at $t$.
In addition, the primal solution (i.e., the algorithm's scheduling decisions) for $t<\tau$ does not change due to the release of $z$.
\item[$(\mathcal{P}3)$ \textbf{Completion}.] $\gamma(t) = 0$ for all $t > C_{\max}^\tau$.
\end{description}
Essentially properties $(\mathcal{P}1)$, $(\mathcal{P}2)$ and $(\mathcal{P}3)$ reflect that the statements of
Lemmas~\ref{lemma:dominancelinear},~\ref{lemma:consistencylinear} and~\ref{lemma:completionlinear} are not tied exclusively
to the total weighted flow time problem.

\begin{theorem} \label{thm:framework1}
Any algorithm that satisfies the properties $(\mathcal{P}1)$, $(\mathcal{P}2)$ and $(\mathcal{P}3)$
with respect to a feasible dual solution is an optimal online algorithm for fractional GFP
and a $(1+\epsilon)$-speed $\frac{1+\epsilon}{\epsilon}$-competitive algorithm for integral GFP.
\end{theorem}
\begin{proof}
The feasibility of the dual solution is directly implied by the fact that $\lambda_j \geq 0$
(since we only increase these dual variables)
and from our definition of $\gamma(t)$ which implies that the constraints~\eqref{DLP1} are satisfied and $\gamma(t) \geq 0$.
Let $C_{\max}$ be the completion time of the last job.
We will assume, without loss of generality, that at time $t \leq C_{\max}$ there is at least one pending job in the schedule;
otherwise, there are idle times in the schedule and we can apply the same type of analysis for jobs scheduled between consecutive idle periods.

We will first show that the primal and the dual objectives are equal.
Consider a job $j$ and let $[t_1,t_2], [t_2,t_3], \ldots, [t_{k-1},t_k]$ be the time intervals during which $j$ is executed.
Note that $x_{j}(t) = 1$ for every $t$ in these intervals (and $x_{j'}(t) = 0$ for $j' \neq j$).
Hence, the contribution of $j$ to the primal (fractional) objective is
\begin{equation*}
\sum_{i=1}^{k-1} \delta_j \int_{t_i}^{t_{i+1}} g(t-r_j) dt
\end{equation*}
By properties~$(\mathcal{P}1)$ and~$(\mathcal{P}2)$, the line $\gamma_j$ is dominant during the same time intervals.
Thus, the contribution of job $j$ to the dual is
\begin{align*}
\lambda_j p_j - \sum_{i=1}^{k-1} \int_{t_i}^{t_{i+1}} \gamma(t) dt
&= \lambda_j p_j - \sum_{i=1}^{k-1} \int_{t_i}^{t_{i+1}} \biggl(\lambda_j - \delta_j g(t-r_j) \biggl) dt \\
&= \sum_{i=1}^{k-1} \delta_j \int_{t_i}^{t_{i+1}} g(t-r_j) dt
\end{align*}
since $\sum_{i=1}^{k-1} \int_{t_i}^{t_{i+1}} \lambda_j dt = \lambda_j \sum_{i=1}^{k-1} \int_{t_i}^{t_{i+1}} x_{j}(t) dt
= \lambda_j p_j$.
The first part of the theorem follows by summing over all jobs $j$, and by accounting for the fact that
$\int_{C_{\max}}^\infty \gamma(t)=0$ (from property $(\mathcal{P}3)$).

For the second part of the theorem, consider again the time intervals during which a job $j$ is executed.
The contribution of $j$ to the integral objective is $w_j g(C_j-r_j) = \delta_j g(C_j-r_j) p_j$.
By properties~$(\mathcal{P}1)$ and~$(\mathcal{P}2)$, for any $t \in \bigcup_{i=1}^{k-1} [t_i,t_{i+1}]$ we have that $\gamma_j(t) \geq 0$.
In particular, it holds for $t_k=C_j$, that is $\lambda_j \geq \delta_j g(C_j-r_j)$.
Therefore, the contribution of $j$ to the integral objective is
\begin{equation*}
w_j g(C_j-r_j) \leq \lambda_j p_j
\end{equation*}
Since we consider the speed augmentation case, we will use as lower bound of the optimal solution the dual program that uses a smaller speed
as explained in Section~\ref{sec:preliminaries}.
By properties~$(\mathcal{P}1)$ and~$(\mathcal{P}2)$, we have $\gamma(t) = \lambda_{j} - \delta_j g(t-r_j) \leq \lambda_j$
during the time intervals where the job $j$ is executed.
Thus, the contribution of $j$ to the dual objective is at least
\begin{equation*}
\lambda_j p_j - \frac{1}{1+\epsilon} \sum_{i=1}^{k-1} \int_{t_i}^{t_{i+1}} \gamma(t) dt
\geq \lambda_j p_j - \frac{1}{1+\epsilon} \sum_{i=1}^{k-1} \int_{t_i}^{t_{i+1}} \lambda_j dt
= \frac{\epsilon}{1+\epsilon} \lambda_j p_j,
\end{equation*}
since $\sum_{i=1}^{k-1} \int_{t_i}^{t_{i+1}} \lambda_j dt = \lambda_j \sum_{i=1}^{k-1} \int_{t_i}^{t_{i+1}} x_{j}(t) dt = \lambda_j p_j$.
From property~$(\mathcal{P}3)$ we have that $\int_{C_{\max}}^\infty \gamma(t) dt =0$.
Summing up over all jobs, the theorem follows.
\end{proof}

In what follows in this section, we apply this framework in three different problems.

\subsection{Online GCP with general cost functions}
\label{subsec:gcp}

In this section, we consider the fractional GCP and we will show that there is an optimal primal-dual algorithm for it.
The following is a linear relaxation of the problem.


\begin{align*}
\min \sum_{j \in \mathcal{J}} \delta_{j} \int_{r_j}^{\infty} & g(t) x_j(t) dt \\
\int_{r_j}^{\infty} x_j(t) dt & \geq p_j & & \forall j \in \mathcal{J} \\
\sum_{j \in \mathcal{J}} x_j(t) & \leq 1 & & \forall t \geq 0 \\
x_j(t) & \geq 0 & & \forall j \in \mathcal{J}, t \geq 0
\end{align*}

The dual program reads
\begin{align*}
\max \sum_{j \in \mathcal{J}} \lambda_j p_j &- \int_0^{\infty} \gamma(t) dt \\
\lambda_j - \gamma(t) & \leq \delta_{j} g(t) & & \forall j \in \mathcal{J}, t \geq r_j \\
\lambda_j &\geq 0 & & \forall j \in \mathcal{J} \\
\gamma(t) &\geq 0 & & \forall t \geq 0
\end{align*}

We define $\gamma_j(t) = \lambda_j - \delta g(t)$, based on the same arguments as in Section 3.
In what follows, we use Procedure~\ref{algo:completion-times} so to define and update the values of $\lambda_j$'s for all jobs in $R_\tau$, i.e., all jobs that are released up to time $\tau$ (completed or not).

\begin{algorithm}[h]
\begin{algorithmic}[1]
\STATE Consider jobs in $R_\tau$ in increasing order of their completion times if no new jobs are released after time $\tau$, i.e., $C_1 < C_2 < \ldots < C_k$
\STATE Choose $\lambda_{k}$ such that $\gamma_k(C_{k})=0$
\FOR{every job $j = k-1$ to $1$}
    \STATE Let $j' \in R_\tau$ be the job scheduled right after $C_j$
    \STATE Choose the $\lambda_j$ such that $\gamma_{j}(C_{j}) = \gamma_{j'}(C_{j})$
\ENDFOR
\end{algorithmic}
\caption{Assignment and updating of $\lambda_{j}$'s for the set $R_\tau$ of all jobs released by time $\tau$ when a new job is released.}
\label{algo:completion-times}
\end{algorithm}

The following lemma shows that the dominance condition is always satisfied for all jobs in $R_{\tau}$ due to Procedure~\ref{algo:completion-times}.

\begin{lemma}
For $\lambda_j$'s as defined by Procedure~\ref{algo:completion-times}, and algorithm $A \equiv HDF$, the properties $(\mathcal{P}1)$ and $(\mathcal{P}2)$ hold.
\end{lemma}
\begin{proof}
Consider the jobs in $R_\tau$ in increasing order of their completion times, i.e., $C_1 < C_2 < \ldots < C_k$.
Let $j$ and $j+1$ be two consecutive jobs in this order.
By the choice of dual variables in Procedure~\ref{algo:completion-times} for a time $t > \max\{r_j,r_{j+1}\}$ we have
\begin{eqnarray*}
\gamma_j(t) & = & \lambda_j - \delta_j g(t) = \lambda_{j+1} - \delta_{j+1} g(C_j) + \delta_j g(C_j) - \delta_j g(t)\\
 & = & \lambda_{j+1} - \delta_{j+1} g(t) + (\delta_{j+1} - \delta_j) (g(t) - g(C_j)) \\
  & = & \gamma_{j+1}(t) + (\delta_{j+1} - \delta_j) (g(t) - g(C_j))
\end{eqnarray*}
As we follow HDF, we have that $\delta_{j+1} - \delta_j \leq 0$.
Since $g$ is a non-decreasing function,
if $t < C_j$ then $\gamma_j(t) \leq \gamma_{j+1}(t)$,
while if $t \geq C_j$ then $\gamma_j(t) \geq \gamma_{j+1}(t)$.
In other words, $\gamma_{j}$ intersects $\gamma_{j+1}$ at $C_{j}$ and that is indeed the unique intersection point between the two curves for every $j,j+1$.
Note that the uniqueness property does not necessarily hold in the settings with flow time objectives, which explains the simplicity of the dual construction for objectives on completion times compared to the ones on flow times.
Then, the dominance property follows using the same arguments as in the proof of Lemma~\ref{lemma:dominancelinear}.

Since the claim holds for every $j$, we deduce that $\gamma_{j}(t) \geq \gamma_{j'}(t)$ for every $j' \in R_\tau$ and every time $t$ during the execution of job $j$.
\end{proof}

Property $(\mathcal{P}3)$ is straightforward from Procedure~\ref{algo:completion-times} and hence we obtain the following theorem.

\begin{theorem} \label{thm:gcp}
The primal-dual algorithm $A \equiv HDF$ is an optimal algorithm for fractional GCP
and a $(1+\epsilon)$-speed $\frac{1+\epsilon}{\epsilon}$-competitive algorithm for integral GCP.
\end{theorem}

\subsection{Online GFP with convex/concave cost functions and equal density jobs}
\label{subsec:convex.same}

In this section, we consider GFP with cost function $g$ that is
either a convex or a concave non-decreasing function with $g(0)=0$.
Moreover, we assume that all jobs have the same density, i.e., $\delta_j=\delta$ for each $j \in \mathcal{J}$.
For both convex and concave cases, we will show that there is an optimal primal-dual algorithm for minimizing the total fractional cost.
For convex functions, this algorithm has to be the FIFO policy, whereas for concave functions the optimal algorithm has to be the LIFO policy.

For both problems, we define $\gamma_j=\lambda_j-\delta g(t-r_j)$ (based on the same arguments as in Section~\ref{section:flow} and the constraints~\eqref{DLP1}).
In what follows, we use Procedures~\ref{algo:linear1} and~\ref{algo:linear2} so to define and update the values of $\lambda_j$'s.

\begin{lemma}\label{lemma:convexconcave1}
For $\lambda_{j}$'s as defined by Procedure~\ref{algo:linear1},
property $(\mathcal{P}1)$ holds if: \\
(i) $g$ is convex, all jobs have equal density and $A \equiv \text{FIFO}$,\\
(ii) $g$ is concave, all jobs have equal density and $A \equiv \text{LIFO}$.
\end{lemma}
\begin{proof}
We will follow the proof of Lemma~\ref{lemma:dominancelinear} by showing
 ($\ref{eq:dominant-ineq-1}$) and ($\ref{eq:dominant-ineq-2}$).
Similarly all jobs have the same density, we obtain
\begin{align*}
\lambda_{j} - \delta g(t-r_{j}) &= \lambda_{j+1} - \delta g(t-r_{j+1}) \\
 & \quad + \delta [(g(C_{j}-r_{j}) - g(C_{j}-r_{j+1})) - (g(t-r_{j}) - g(t-r_{j+1}))]
\end{align*}

\noindent (i) Suppose that $g$ is convex.
Since $C_j < C_{j+1}$ and the scheduling algorithm is FIFO, we have that $r_j \leq r_{j+1}$.
Thus, $g(C_{j}-r_{j}) \geq g(C_{j}-r_{j+1})$ and $g(t-r_{j}) \geq g(t-r_{j+1})$, since $g$ is non-decreasing.
If $t \geq C_j$, then by convexity it follows that $g(C_{j}-r_{j}) - g(C_{j}-r_{j+1}) \leq g(t-r_{j}) - g(t-r_{j+1})$,
and hence ($\ref{eq:dominant-ineq-1}$) holds.
If $t < C_j$, then by convexity it holds that $g(C_{j}-r_{j}) - g(C_{j}-r_{j+1}) \geq g(t-r_{j}) - g(t-r_{j+1})$,
and hence ($\ref{eq:dominant-ineq-2}$) holds.
\medskip

\noindent (ii) Suppose that $g$ is concave.
Since $C_j < C_{j+1}$ and the scheduling algorithm is LIFO, we have that $r_j \geq r_{j+1}$.
Thus, $g(C_{j}-r_{j}) \leq g(C_{j}-r_{j+1})$ and $g(t-r_{j}) \leq g(t-r_{j+1})$, since $g$ is non-decreasing.
If $t \geq C_j$, then by concavity it follows that $g(C_{j}-r_{j}) - g(C_{j}-r_{j+1}) \leq g(t-r_{j}) - g(t-r_{j+1})$,
and hence ($\ref{eq:dominant-ineq-1}$) holds.
If $t < C_j$, then by concavity it holds that $g(C_{j}-r_{j}) - g(C_{j}-r_{j+1}) \geq g(t-r_{j}) - g(t-r_{j+1})$,
and hence ($\ref{eq:dominant-ineq-2}$) holds.
\end{proof}

\begin{lemma}\label{lemma:convexconcave2}
For $\lambda_{j}$'s defined by Procedures~\ref{algo:linear1} and~\ref{algo:linear2},
property $(\mathcal{P}2)$ holds if \\
(i) $g$ is convex, all jobs have equal density and $A \equiv \text{FIFO}$,\\
(ii) $g$ is concave, all jobs have equal density and $A \equiv \text{LIFO}$.
\end{lemma}
\begin{proof}
We rely on the following claim. Recall that $P_\tau$ denotes the set of
the algorithm's pending jobs at the release of $z$ at time $\tau$.
\begin{claim}\label{claim:convexconcave-change}
Let $j_{1}$ and $j_{2}$ be two jobs in $P_\tau$ such that $C_{j_1} < C_{j_2}$.
Then $\Delta_{j_{1}} \geq \Delta_{j_{2}}$.
\end{claim}
\begin{claimproof}~\\
\noindent(i) Suppose that $g$ is convex and we follow the FIFO policy, then
$C_{j_1} <C_{j_2}<C_z$, In this case, both $\gamma_{j_1}$ and $\gamma_{j_2}$
are moved up by $\gamma_{z}(C_{z} - p_{z} - r_{z}) - \gamma_{z}(C_{z} - r_{z})$.
Therefore, $\Delta_{j_{1}} = \Delta_{j_{2}} = \delta (g(C_{z} - r_{z}) - g(C_{z} - p_{z} - r_{z}))$.
\medskip

\noindent(ii) Suppose that $g$ is concave and we follow the LIFO policy.
Then $C_z < C_{j_1} <C_{j_2}$.
Hence, the completion times of both $j_{1}$ and $j_{2}$ are delayed by $p_{z}$ by the algorithm.
In other words, before the release of $z$ the completion times of $j_{1}$ and $j_{2}$ were
$C_{j_{1}}-p_{z}$ and $C_{j_{2}}-p_{z}$, respectively.
Moreover, the relative ordering of all jobs in $P_\tau \setminus \{z\}$ (as done by the algorithm)
remains the same before and after the release of $z$.
Thus, by Procedure~\ref{algo:linear1}, we have
$\lambda_{j_{1}} - \delta g(C_{j_{1}} - r_{j_{1}}) = \lambda'_{j_{1}} - \delta g(C_{j_{1}} - p_{z} - r_{j_{1}})$,
where $\lambda'_{j_{1}}$ is the value of the dual variable of $j_{1}$ before the release of $z$.
Hence, $\Delta_{j_{1}} = \delta(g(C_{j_{1}} - r_{j_{1}}) - g(C_{j_{1}} - p_{z} - r_{j_{1}}))$.
Similarly, $\Delta_{j_{2}} = \delta(g(C_{j_{2}} - r_{j_{2}}) - g(C_{j_{2}} - p_{z} - r_{j_{2}}))$.
Therefore, $\Delta_{j_{1}} \geq \Delta_{j_{2}}$ since $g$ is concave and by the LIFO algorithm $r_{j_1} \geq r_{j_2}$.
\end{claimproof}

Note that Claims~\ref{claim:linear-order} and~\ref{claim:linear-continuous}
also hold for the problems we study in this section. Therefore, by
applying the same arguments as in the proof of Lemma~\ref{lemma:consistencylinear}, and by replacing
Claim~\ref{claim:linear-change} with Claim~\ref{claim:convexconcave-change}, we arrive at the same conclusion.
\end{proof}

The proof of the following lemma directly follows from Lemmas~\ref{lemma:convexconcave1} and~\ref{lemma:convexconcave2} as in the linear case.
\begin{lemma}\label{lemma:convexconcave3}
For $\lambda_{j}$'s as defined by Procedures~\ref{algo:linear1} and~\ref{algo:linear2},
the property $(\mathcal{P}3)$ holds if
(i) $g$ is convex and all jobs have equal density; or
(ii) $g$ is concave and all jobs have equal density.
\end{lemma}

By Lemmas~\ref{lemma:convexconcave1},~\ref{lemma:convexconcave2} and~\ref{lemma:convexconcave3},
the properties $(P_1)$, $(P_2)$ and $(P_3)$, respectively, are satisfied, and hence we obtain the following theorem.

\begin{theorem} \label{thm:convex-concave}
The primal-dual algorithm $A \equiv FIFO$ (resp. $A \equiv LIFO$) is an optimal online algorithm for fractional GFP
and a $(1+\epsilon)$-speed $\frac{1+\epsilon}{\epsilon}$-competitive algorithm for integral GFP,
when we consider convex (reps. concave) cost functions and jobs of equal density.
\end{theorem}

\subsection{PSP with positive constraint coefficients}
\label{subsec:psp}

In this section we study the PSP problem (defined formally in Section~\ref{sec:introduction}), assuming
constraints of the form $Bx \leq 1$, and $b_{ij} > 0$ for every $i,j$.
We denote by $b_j$, $B_j$ the smallest and largest element of each column of $B$, respectively,
i.e., $b_{j} = \min_{i} b_{ij}$ and $B_{j} = \max_{i} b_{ij}$.
The primal LP relaxation of the problem, and its dual are as follows:

\begin{align*}
\min \sum_{j \in \mathcal{J}} \delta_{j} \int_{r_j}^{\infty} & (t-r_j) x_j(t) dt \nonumber\\
\int_{r_j}^{\infty} x_j(t) dt &\geq p_j \qquad \forall j \in \mathcal{J} \\
\sum_{j \in \mathcal{J}} b_{ij}x_j(t) &\leq 1 \qquad \forall t \geq 0, \forall 1 \leq i \leq m \\
x_j(t) &\geq 0 \qquad \forall j \in \mathcal{J}, t \geq 0 \notag 
\end{align*}

\begin{align*}
\max \sum_{j \in \mathcal{J}} \lambda_j p_j - \sum_{i=1}^{m} & \int_0^{\infty} \gamma_{i}(t) dt \nonumber\\
\lambda_j - \sum_{i} b_{ij} \gamma_{i}(t) \leq \delta_{j} & (t-r_j) \quad \forall j \in \mathcal{J}, t \geq r_{j} \\
\lambda_j \geq 0 & \quad \forall j \in \mathcal{J} \nonumber\\
\gamma_{i}(t) \geq 0 & \quad \forall t \geq 0 \nonumber
\end{align*}

We begin with the intuition behind the analysis, and how one can exploit the ideas of the analysis
of HDF of Section~\ref{section:flow}, in the context of minimum total flow.
The essential difference between the two problems is the set of packing constraints that are present in the PSP formulation.
This difference manifests itself in the dual with the constraints $\lambda_j - \sum_{i} b_{ij} \gamma_{i}(t) \leq \delta_{j} (t-r_j)$,
for all $j \in \mathcal{J}, t \geq 0$.
In contrast, the dual LP of the minimum total weighted flow time problem has corresponding constraints $\lambda_j - \gamma(t) \leq \delta_{j} (t-r_j)$.
At this point, one would be motivated to define $\mu(t) = \sum_{i=1}^m b_{ij}\gamma_i(t)$, and then view this $\mu(t)$ as the variable $\gamma(t)$
in the dual LP formulation of the minimum total flow problem (and thus one would proceed as in the analysis of HDF).
However, a complication arises: it is not clear how to assign the variables $\gamma_i(t)$, for given $\mu(t)$.
Because of this, we follow a different way (which will guarantee the feasibility of the $\gamma_i(t)$'s'):
Instead of satisfying the constraint $\lambda_j - \sum_{i=1}^m b_{ij} \gamma_{i}(t) \leq \delta_{j} (t-r_j)$
we will satisfy the stronger constraint $\lambda_j -b_j \sum_{i=1}^m \gamma_i(t) \leq \delta_{j} (t-r_j)$.
At an intuitive level, we will satisfy the constraints $\lambda'_j - \mu'(t) \leq \delta'_{j} (t-r_j)$,
where $\delta'_{j} = \delta_{j}/b_{j}$ and $\mu'(t)$ and $\lambda'_{j}$ are variables.
Using the same scheme as in the setting of linear functions we can construct $\lambda'_{j}$ (so $\mu'_{j}(t)$),
so that the properties $(\mathcal{P}1), (\mathcal{P}2), (\mathcal{P}3)$ are satisfied.
Particularly, if job $j$ is processed at time $t$ then $\mu'_{j}$ is dominant at $t$, i.e.,
$\mu'_{j}(t) = \lambda'_j-\delta'_j(t-r_j)=\mu'(t) = \max_{k} \mu'_{k}(t)$.
Moreover, $\lambda'_{j} - \mu'(t) \leq \delta'_{j} (t-r_{j})$ for every job $j$ and $t \geq r_{j}$.

As a last step, we need to translate the above dual variables to the PSP problem.
Define $\lambda_{j} = \lambda'_{j} b_{j}$ for every job $j$.
Let $j$ be the job processed at time $t$ and the constraint $i$ is a tight constraint at time $t$, i.e., $b_{ij} = B_{j}$.
Set $\gamma_{i}(t) = \mu'(t)$ and $\gamma_{i'}(t) = 0$ for $i' \neq i$.

\paragraph{Algorithm.} The above discussion implies an adaptation of HDF to the PSP problem as follows:
At any time $t$, we schedule job the job $j$ which attains the highest ratio
$\delta_{j}/b_{j}$ among all pending jobs, at rate $x_{j}(t) = 1/B_{j}$.

Based on this algorithm, we prove Theorem~\ref{thm:psp} using similar arguments as for the proof of Theorem~\ref{thm:linear}.

\begin{theorem} \label{thm:psp}
For the online PSP problem with constraints $Bx \leq 1$ and $b_{ij} > 0$ for every $i,j$,
an adaptation of HDF is $\max_{j} \{B_{j}/b_{j}\}$-speed 1-competitive for fractional weighted flow time
and $\max_{j} \{(1+\epsilon)B_{j}/b_{j}\}$-speed $(1+\epsilon)/\epsilon$-competitive for integral weighted flow time.
\end{theorem}
\begin{proof}
We first show the feasibility of the dual solution.
For every job $j$ and time $t$, we have
\begin{align*}
\lambda_{j}/b_{j}- \delta_{j}/b_{j} \cdot (t-r_{j}) = \mu'_{j}(t)
\leq \mu'(t) = \sum_{i} \gamma_{i}(t) \leq \sum_{i} b_{ij}/b_{j} \cdot \gamma_{i}(t)
\end{align*}
where the last inequality is due to $b_{ij} \geq \min_{i}b_{ij} = b_{j} > 0$.

We will now bound the integral and fractional primal cost (with unit speed) by the dual cost (with speed $b_{j}/B_{j}$).
Consider a job $j$ and let $[t_1,t_2], [t_2,t_3],$ $ \ldots, [t_{m-1},t_m]$ be the time intervals
during which $j$ is executed. Note that $x_{j}(t) = 1/B_{j}$ for every $t$ during the intervals
(and $x_{j'}(t) = 0$ for $j' \neq j$).
The contribution of job $j$ to the fractional primal objective is
\begin{equation*}
\sum_{a=1}^{m-1} \int_{t_a}^{t_{a+1}} \delta_j (t-r_j) x_{j}(t) dt =
\sum_{a=1}^{m-1} \int_{t_a}^{t_{a+1}} \delta'_{j} (t-r_{j}) b_{j}/B_{j} dt
\end{equation*}

The contribution of $j$ to the integral primal objective is
\begin{equation*}
w_{j} (C_{j}-r_{j}) = p_{j} \delta_{j}(C_{j}-r_{j}) \leq p_{j} \lambda_{j}
\end{equation*}
since the curve $\lambda'_{j} - \delta'_{j}(t-r_{j}) \geq 0$ for every $t$
during which $j$ is executed, particularly for $t = C_{j}$ (and note that
$\lambda_{j} = b_{j} \lambda'_{j}$ and $\delta_{j} = b_{j} \delta'_{j}$).

Assuming that job $j$ is processed on the machine with speed $b_{j}/B_{j}$,
the contribution of job $j$ to the dual is
\begin{align*}
\lambda_j p_j &- \frac{b_{j}}{B_{j}} \sum_{a=1}^{m-1} \int_{t_a}^{t_{a+1}} \sum_{i}\gamma_{i}(t) dt
= \sum_{a=1}^{m-1} \int_{t_a}^{t_{a+1}} \biggl( \lambda_j x_{j}(t) - \frac{b_{j}}{B_{j}} \sum_{i}\gamma_i(t) \biggl) dt \\
&= \sum_{a=1}^{m-1} \int_{t_a}^{t_{a+1}} \frac{b_{j}}{B_{j}} \biggl( \lambda'_j - \mu'(t) \biggl) dt
= \sum_{a=1}^{m-1} \int_{t_a}^{t_{a+1}} \frac{b_{j}}{B_{j}} \delta'_{j} (t-r_{j}) dt
\end{align*}
where the first equality is due to $\sum_{a=1}^{m-1} \int_{t_a}^{t_{a+1}}= p_{j}$; the
second equality follows by $x_{j}(t) = 1/B_{j}$ and the definitions of dual variables;
and the last equality holds by the dominance property. Hence,
the contribution of job $j$ to the fractional primal cost and that to the dual one (with speed $b_{j}/B_{j}$)
are equal. As the latter holds for every job, the first statement follows.

We can bound differently the contribution of $j$ to the dual.
\begin{align*}
\lambda_j p_j &- \frac{b_{j}}{(1+\epsilon)B_{j}} \sum_{a=1}^{m-1} \int_{t_a}^{t_{a+1}} \sum_{i}\gamma_{i}(t) dt
\geq \lambda_j p_j - \frac{1}{(1+\epsilon)B_{j}} \sum_{a=1}^{m-1} \int_{t_a}^{t_{a+1}} \lambda_{j} dt \\
&= \lambda_j p_j - \frac{\lambda_{j}}{1+\epsilon} \sum_{a=1}^{m-1} \int_{t_a}^{t_{a+1}} x_{j}(t) dt
= \frac{\epsilon}{1+\epsilon} \lambda_{j} p_{j}
\end{align*}
where the first inequality due to $\lambda_{j}/b_{j} = \lambda'_{j} \geq \mu'(t) = \sum_{i} \gamma_{i}(t)$
for every $t$ during which $j$ is processed; the first equality holds since job $j$ is processed at rate $x_{j}(t) = 1/B_{j}$;
and the last equality follows $\int x_{j}(t)dt = p_{j}$. Therefore, the contribution of job $j$ to the integral primal cost
is at most $(1 + \epsilon)/\epsilon$ that to the dual objective. Again, summing over all jobs, the second statement holds.
\end{proof}

\section{A generalized framework using dual-fitting}
\label{sec:competitive}

In this section we relax certain properties as established in Section~\ref{section:framework-opt}
in order to generalize our framework and apply it to the integral variant of more problems.
Our analysis here is based on the dual-fitting paradigm, since the analysis of Section~\ref{section:flow}
provides us with intuition about the geometric interpretation of the primal and dual objectives.
We consider, as concrete applications, GFP for given cost functions $g$.
We again associate with each job $j$ the curve $\gamma_j$ and set $\gamma(t)=\max\{0,\max_{j \in \mathcal{J}: r_j \leq t}\{\gamma_j(t)\}\}$.
Then, we need to define how to update the dual variables $\lambda_j$, upon release of a new job $z$ at time $\tau$,
such that the following properties are satisfied:
\begin{description}
\item[$(\mathcal{Q}1)$] If the algorithm $A$ schedules job $j$ at time $t\geq \tau$
then $\gamma_{j}(t) \geq 0$ and $\lambda_{j} \geq \gamma_{j'}(t)$ for every other pending job $j'$ at time $t$.
\item[$(\mathcal{Q}2)$] If the algorithm $A$ schedules job $j$ at time $t<\tau$,
then $\gamma_{j}(t) \geq 0$ and $\lambda_{j} \geq \gamma_{j'}(t)$ for every other pending job $j'$ at time $t$.
In addition, the primal solution for $t<\tau$ is not affected by the release of $z$.
\item[$(\mathcal{Q}3)$] $\gamma(t) = 0$ for all $t > C_{\max}^\tau$.
\end{description}

Note that $(\mathcal{Q}1)$ is relaxed with respect to property $(\mathcal{P}1)$ of Section~\ref{section:framework-opt},
since it describes a weaker dominance condition.
Informally, $(\mathcal{Q}1)$ guarantees that for any time $t$ the job that is scheduled at $t$ does not have negative contributions in the dual.
On the other hand, property~$(\mathcal{Q}2)$ is the counterpart of $(\mathcal{Q}1)$,
for times $t<\tau$ (similar to the relation between $(\mathcal{P}1)$ and $(\mathcal{P}2)$).
Finally, note that even though the relaxed properties do not  guarantee anymore the optimality for the fractional objectives,
the following theorem (Theorem~\ref{thm:framework2}) establishes exactly the same result as Theorem~\ref{thm:framework1} for the {\em integral} objectives.
This is because in the second part of the proof of Theorem~\ref{thm:framework1} we only require that
when $j$ is executed at time $t$ then $\lambda_{j} \geq \gamma_{j'}(t)$ for every other pending job $j'$ at $t$,
which is in fact guaranteed by properties~$(\mathcal{Q}1)$ and~$(\mathcal{Q}2)$.

\begin{theorem} \label{thm:framework2}
Any algorithm that satisfies the properties $(\mathcal{Q}1)$, $(\mathcal{Q}2)$ and $(\mathcal{Q}3)$ with respect to a feasible dual solution
is a $(1+\epsilon)$-speed $\frac{1+\epsilon}{\epsilon}$-competitive algorithm for integral GFP with general cost functions $g$.
\end{theorem}
\begin{proof}
The feasibility of the dual solution is straightforward.
We denote by $C_{\max}$ the completion time of the last job.
As in the proof of Theorem~\ref{thm:framework1}, we can assume, without loss of generality, that at time $t \leq C_{\max}$ there is at least one pending job in the schedule.

Consider now a job $j$ and let $[t_1,t_2], [t_2,t_3], \ldots, [t_{k-1},t_k]$ be the time intervals during which $j$ is executed.
Note that $x_{j}(t) = 1$ for every $t$ during the intervals (and $x_{j'}(t) = 0$ for $j' \neq j$).
Hence, the contribution of $j$ to the primal objective is
\begin{equation*}
\sum_{i=1}^{k-1} \int_{t_i}^{t_{i+1}} \delta_j g(t-r_j) dt
    \leq \sum_{i=1}^{k-1} \int_{t_i}^{t_{i+1}} \lambda_{j} dt
    = \lambda_{j} p_{j}
\end{equation*}
where the inequality follows from properties $(\mathcal{Q}1)$ and $(\mathcal{Q}2)$, and specifically due to the constraint that $\gamma_j(t) \geq 0$.

From properties $(\mathcal{Q}1)$ and $(\mathcal{Q}2)$, we have $\lambda_{j} \geq \gamma(t)$ during the same time intervals.
Since we assume that the optimal solution uses a speed of $\frac{1}{1+\epsilon}$, the contribution of job $j$ to the dual objective is at least
\begin{equation*}
\lambda_j p_j - \frac{1}{1+\epsilon} \sum_{i=1}^{k-1} \int_{t_i}^{t_{i+1}} \gamma(t) dt
\geq \lambda_j p_j - \frac{1}{1+\epsilon} \sum_{i=1}^{k-1} \int_{t_i}^{t_{i+1}} \lambda_j dt
= \frac{\epsilon}{1+\epsilon} \lambda_j p_j
\end{equation*}
In addition, from  property ($\mathcal{Q}3)$ we have $\int_{C_{\max}}^\infty \gamma(t) dt =0$.
Summing up over all jobs $j$, the theorem follows.
\end{proof}

\subsection{Online GFP with general cost functions and equal-density jobs}
\label{subsec:equal-density}

We will analyze the FIFO algorithm using dual fitting.
We will use a single procedure, namely Procedure~\ref{algo:general-equaldensity},
for the assignment of the $\lambda_j$ variables for each job $j$ released by time $\tau$.
We denote this set of jobs by $R_\tau$, and $k=|R_\tau|$.

\begin{algorithm}[h]
\begin{algorithmic}[1]
\STATE Consider jobs in $R_\tau$ in increasing order of completion times if no new jobs are released after time $\tau$, i.e., $C_1 < C_2 < \ldots < C_k$
\STATE Choose $\lambda_k$ such that $\gamma_k(C_{k})=0$
\FOR{$j = k-1$ to $1$}
    \STATE Choose the maximum possible $\lambda_j$ such that for every $t \geq C_j$, $\gamma_j(t) \leq \gamma_{j+1}(t)$
    \IF {$\gamma_j(C_j)<0$}
        \STATE Choose $\lambda_j$ such that $\gamma_j(C_j)=0$
    \ENDIF
\ENDFOR
\end{algorithmic}
\caption{Assignment and updating of $\lambda_{j}$'s for the set $R_\tau$ of all jobs released by time $\tau$ when a new job is released.}
\label{algo:general-equaldensity}
\end{algorithm}

We will need first the following simple proposition:

\begin{proposition}\label{lemma:general-equaldensity}
Let $i,j \in \mathcal{J}$ be two jobs such that $r_i \leq r_j$ and suppose that there is a time $t_0 \geq r_j$ such that $\gamma_i(t_0) \geq \gamma_j(t_0)$.
Then, $\lambda_i \geq \lambda_j$.
\end{proposition}
\begin{proof}
By the statement of the proposition we have that $\lambda_i - \delta g(t_0 - r_i) \geq \lambda_j - \delta g(t_0 - r_j)$.
Since the function $g$ is non-decreasing and $r_i \leq r_j$, it must be that $\lambda_i \geq \lambda_j$.
\end{proof}

The following lemma is instrumental in establishing the desired properties.

\begin{lemma} \label{lemma:qproperties}
For every job $j$ in $R_\tau$, $\lambda_j \geq \gamma_i(C_{j-1})$.
\end{lemma}
\begin{proof}
Consider the jobs in $R_\tau$ in increasing order of completion times.
We will prove the lemma by considering two cases: for jobs $i < j$, and for jobs $i > j$ (according to the above ordering).
In both cases we apply induction on $i$.

\medskip
\noindent
{\em Case 1: $i \leq j$.} \ For base case $i=j$, the inequality becomes $\lambda_j \geq \gamma_j(C_{j-1})$
which trivially holds since $\lambda_{j} \geq \gamma_{j}(t)$ for all $t$.

For the induction step, suppose that the lemma is true for a job $i < j$, that is
$\lambda_j \geq \gamma_i(C_{j-1})$.
Recall that job $i-1 \in R_\tau$ is the last job that is completed
before $i$. If $\lambda_{i-1}$ had not been set in line~6 of
Procedure~\ref{algo:general-equaldensity} then it has been set in line~4 of the procedure; thus,
from the induction hypothesis,
$\gamma_{i-1}(C_{i-1}) \leq \gamma_i(C_{i-1}) \leq \lambda_j$. On the other hand, if
$\lambda_{i-1}$ is set in line~6 of
Procedure~\ref{algo:general-equaldensity} then
$\gamma_{i-1}(C_{i-1})=0$. Hence $\lambda_j > 0 =
\gamma_{i-1}(C_{i-1}) \geq \gamma_{i-1}(C_{j-1})$
where the last inequality is due to the fact that $\gamma_{i-1}$ is a
non-increasing function.

\medskip
\noindent
{\em Case 2: $i \geq j$.} \
We will show the stronger claim that $\lambda_{j} \geq \lambda_{i}$ for every $i \geq j$.
This suffices since $\lambda_{i} \geq \gamma_{i}(t)$ for every $t$.
The base case in the induction, namely the case $i = j$ is trivial.

Assume that the claim is true for a job $i > j$, that is
$\lambda_j \geq \lambda_{i}$.
We need to show that $\lambda_j \geq \gamma_{i+1}(C_{j-1})$.
If $\lambda_i$ is not set in line~6 of
Procedure~\ref{algo:general-equaldensity} (and thus set in line~4 instead), then there is time
$t_0$ such that $\gamma_i(t_0) = \gamma_{i+1}(t_0)$ by the
maximality in the choice $\lambda_j$.
On the other hand, if $\lambda_i$ is set in line~6
of Procedure~\ref{algo:general-equaldensity} then there is at
least one time $t_1$, with $C_{j} \leq t_{1} \leq t_{0}$ such that
$\gamma_i(t_1) \geq \gamma_{i+1}(t_1)$.
In both cases, from Proposition~\ref{lemma:general-equaldensity}
we deduce that $\lambda_i \geq \lambda_{i+1}$. Moreover, $\lambda_{i} \leq \lambda_{j}$
by the induction hypothesis. The claim follows.
\end{proof}

\begin{lemma} \label{lemma:properties.general.equal}
Procedure~\ref{algo:general-equaldensity} satisfies properties $(\mathcal{Q}1)$, $(\mathcal{Q}2)$ and $(\mathcal{Q}3)$.
\end{lemma}
\begin{proof}
From Lemma~\ref{lemma:qproperties} it follows that properties $(\mathcal{Q}1)$ and $(\mathcal{Q}2)$ are satisfied.
More precisely, observe that by line 6 of Procedure~\ref{algo:general-equaldensity} we have
$\gamma_j(t)\geq 0$ for any $t \leq C_j$. It remains to show that for any
time $t$ at which a job $j$ is executed, we have
$\lambda_j \geq \gamma(t) = \max_{i \in \mathcal{J}}\gamma_i(t)$. Since $g$
is non-decreasing, it follows that $\gamma_j(t)$ is non-increasing, hence it
suffices to show the above only for $t=C_{j-1}$; this is indeed established by Lemma~\ref{lemma:qproperties}.

Moreover, concerning property $(\mathcal{Q}3)$, we observe that for all jobs $j \in R_\tau$ for which
$\lambda_j$ is set in line~6 of Procedure~\ref{algo:general-equaldensity}, the property follows straightforwardly.
For all remaining jobs in $R_\tau$, the property holds using the same arguments as in Lemma~\ref{lemma:completionlinear}
for the minimum flow time problem.
\end{proof}

The above lemma in conjunction with Theorem~\ref{thm:framework2} lead to the following result.

\begin{theorem} \label{thm:equal}
FIFO is a $(1+\epsilon)$-speed $\frac{1+\epsilon}{\epsilon}$-competitive for integral GFP with general cost functions and equal-density jobs.
\end{theorem}

\subsection{Online GFP with concave cost functions}
\label{subsec:concave}

We will analyze the HDF algorithm using dual fitting.
As in Section~\ref{section:flow}, we will employ two procedures for maintaining the dual variables $\lambda_j$.
The first one is Procedure~\ref{algo:energy} which updates the $\lambda_j$'s for $j \in P_\tau$.
The second procedure updates the $\lambda_j$'s for $j \in {\cal J}_\tau$;
this procedure is identical to Procedure~\ref{algo:linear2} of Section~\ref{section:flow}.

The intuition behind Procedure~\ref{algo:energy} is to ensure property $({\cal Q}1)$ that is,
$\gamma_j(t) \geq 0$ and $\lambda_j \geq \gamma_{j'}(t)$ for all $j' \in P_\tau$, which in some sense is the ``hard'' property to maintain.
Specifically, for given job $j$ there is a set of jobs $A$ (initialized in line~4) for which the property does not hold.
The while loop in the procedure decreases the $\lambda$ values of jobs in $A$ so as to rectify this situation (see line~6(ii)).
However, this decrement may, in turn, invalidate this property for some jobs $b$ (see line~6(i)).
These jobs are then added in the set of ``problematic'' jobs $A$ and we continue until no problematic jobs are left.
One can formally argue that this procedure terminates.

\begin{algorithm}[htp]
\begin{algorithmic}[1]
\STATE{Consider the jobs in $P_\tau$ in increasing order of completion times if no new jobs are released after time $\tau$, i.e., $C_1<C_2<\ldots <C_k$} 
\STATE For every $1 \leq j \leq k$ choose $\lambda_{j}$ such that $\gamma_{j}(C_{k}) = 0$
\FOR{$j = 2$ to $k$}
\STATE{Define $A := \{\text{jobs }1 \leq a \leq j-1 : \gamma_{a}(C_{j-1}) > \lambda_{j}\}$ }
    \WHILE{$A \neq \emptyset$}
        \STATE Continuously reduce $\lambda_{a}$ by the same amount for all jobs $a \in A$ until: \\
        (i) $\exists$  $a \in A$ and $b \in P_\tau \setminus A$ with $b<a$ such that
        $\lambda_{a} = \gamma_{b}(C_{a-1})$;
        then
        $A \gets A \cup \{b\}$  \\
        (ii) $\exists$  $a \in A$ such that $\gamma_{a}(C_{j-1}) = \lambda_{j}$;
        then $A \leftarrow A \setminus \{a\} $
    \ENDWHILE
\ENDFOR
\end{algorithmic}
\caption{Assignment of dual variables $\lambda_{j}$ for all $j \in P_\tau$ at the arrival of a new job.}
\label{algo:energy}
\end{algorithm}

We will first need to show the following technical lemma:

\begin{lemma} \label{lem:concave-density-high}
Let $a, b \in {\cal J}$ such that $\delta_{a} \geq \delta_{b}$ and suppose that there is a time $t_0 \geq \max\{r_{a},r_{b}\}$
such that $\gamma_{a}(t_{0}) \geq \gamma_{b}(t_{0})$.
Then $\lambda_{a} \geq \gamma_{b}(t)$ for every $t \geq \max\{r_{a},r_{b}\}$.
\end{lemma}
\begin{proof}
By assumption, we have
$
\lambda_{a} - \delta_{a}g(t_{0}-r_{a}) \geq \lambda_{b} - \delta_{b}g(t_{0}-r_{b})
$.
If $r_{a} \leq r_{b}$ then $g(t_0-r_{a}) \geq g(t_0-r_{b})$ so $\delta_{a}g(t_0-r_{a}) \geq \delta_{b}g(t_0-r_{b})$.
Hence $\lambda_{a} \geq \lambda_{b} \geq \gamma_{b}(t)$ for
every $t$. Remains then to consider the case $r_{a} > r_{b}$. Since
$\gamma_{b}(t)$ is non-increasing, suffices to prove that
$\lambda_{a} \geq \gamma_{b}(r_{a})$. Since
$\lambda_{a} \geq \lambda_{b} - \delta_{b}g(t_0-r_{b}) + \delta_{a}g(t_0-r_{a})$, it will suffice to show that
\begin{align*}
\lambda_{b} - \delta_{b}g(t_0-r_{b}) + \delta_{a}g(t_0-r_{a}) &\geq \lambda_{b} - \delta_{b} g(r_{a}-r_{b}) \\
\Leftrightarrow \qquad \delta_{a} g(t_0-r_{a}) &\geq \delta_{b}( g(t_0 - r_{b}) - g(r_{a} - r_{b}))
\end{align*}
Note that $g$ is concave and $g(0) = 0$ so $g$ is sub-additive (sub-linear). Therefore,
the right-hand side is upper bounded by $\delta_{b} g(t_0-r_{a})$, which is at most
the left-hand side.
\end{proof}

The following lemma is related to Procedure~\ref{algo:energy}, which rectifies property $(\mathcal{Q}1)$ iteratively.
It is not hard to see that this procedure terminates, since whenever a job $a$ is added to the set $A$, then $\lambda_a$ is decreased.
Note that a job may be added and removed several times however, $\lambda_a$ cannot be smaller than zero.

\begin{lemma}
\label{lemma:concave.cases}
At the end of the for-loop for job $j$ (line~3) of Procedure~\ref{algo:energy}:\\
(i) For any two jobs $a,b \in P_\tau$ such that $1 \leq a < b \leq j$ (i.e, $\delta_{a} \geq \delta_{b}$),
    we have $\lambda_{a} \geq \gamma_{b}(t)$ for all $t \in [C_{a-1},C_{a}]$; and\\
(ii) For any two jobs $a,b \in P_\tau$ such that $1 \leq b < a \leq j$ (i.e,  $\delta_{a} \leq \delta_{b}$),
     we have $\lambda_{a} \geq \gamma_{b}(t)$ for all  $t \in [C_{a-1},C_{a}]$.
\end{lemma}
\begin{proof}~\\
\noindent
{\em (i)} \ Proof by induction on $j$.
For the base case ($j=1$) note that in line~2 of the procedure, $\gamma_{j}(C_k) = 0$ for every $j \in P_\tau$;
hence from Lemma \ref{lem:concave-density-high} the base case is satisfied.
Suppose that the claim holds at the end of the for-loop for job $j-1$; we will call this for-loop the {\em iteration} for job $j-1$.
Consider the iteration of job $j$. Suppose that during this iteration, $\lambda_{a}$ has decreased more than $\lambda_{b}$ has,
since otherwise the claim follows directly from the induction hypothesis.
This implies that at some moment during the execution of the while loop,
$a \in A$ and $b \notin A$. At that moment, we deduce that $\gamma_{b}(C_{j-1}) \leq \lambda_{j} < \gamma_{a}(C_{j-1})$.
However, by line~6(ii), $\lambda_{a}$ stops decreasing at the point in which $\lambda_{j} = \gamma_{a}(C_{j-1})$.
Therefore, at the end of the iteration of job $j$, we have $\gamma_{b}(C_{j-1}) \leq \gamma_{a}(C_{j-1})$.
Applying Lemma \ref{lem:concave-density-high} by choosing $t_{0} = C_{j-1}$,
it holds that $\lambda_{a} \geq \gamma_{b}(t)$ for $t \in [C_{a-1},C_{a}]$.
\medskip

\noindent
{\em (ii)} The proof is again by induction on $j$.
The base case, $j=1$, holds since no other jobs in $P_\tau$ have higher density than job 1.
Assume that the statement holds at the end of the for-loop for job $j-1$.
During the iteration of the for-loop for job $j$, a set $A$ contains
jobs $a < j$ such that $\lambda_{j} < \gamma_{a}(C_{j-1})$. By the procedure,
$\lambda_{a}$ for every $a \in A$ is decreased until $\lambda_{j} = \gamma_{a}(C_{j-1})$.
So at the end of the iteration, $\lambda_{j} \geq \gamma_{a}(C_{j-1})$ for every $a < j$.
It remains to show that at the end of the iteration, $\lambda_{a} \geq \gamma_{b}(C_{a-1})$
for jobs $b < a < j$. Let $b < a < j$ be two arbitrary jobs. By the induction hypothesis,
$\lambda_{a} \geq \gamma_{b}(C_{a-1})$ holds before the iteration.
During the iteration, if $\lambda_{a}$ is not modified then the inequality remains true
(since $\lambda_{b}$ is not increased). Otherwise, $a$ must be added to $A$ at some moment.
As $\lambda_{a}$ is decreased so probably at some later moment $\lambda_{a} = \gamma_{b}(C_{a-1})$.
However, job $b$ will be added to $A$ and the $\lambda$-values of both jobs
will be decreased by the same amount. Therefore, at the end of the iteration
$\lambda_{a} \geq \gamma_{b}(C_{a-1})$ (the induction step is done).
\end{proof}

\begin{lemma} \label{lem:concave-density-small}
Procedures~\ref{algo:energy} and~\ref{algo:linear2} combined satisfy properties $(\mathcal{Q}1)$,$(\mathcal{Q}2)$ and $(\mathcal{Q}3)$.
\end{lemma}
\begin{proof}
We will first show that at the end of Procedure~\ref{algo:energy} the dual solution satisfies property $(\mathcal{Q}1)$.
From Lemma~\ref{lemma:concave.cases}, given $j \in P_\tau$,
we have $\lambda_{j} \geq \gamma_{a}(t)$ for every job $a \in P_\tau$ and for every time $t \in [C_{j-1},C_{j}]$.
It remains to show that $\gamma_{j}(t) \geq 0$ for $t \in [C_{j-1},C_{j}]$.
Since $\gamma_j$ is non-increasing, suffices to show that $\gamma_{j}(C_{j}) \geq 0$.
After line~2 of the procedure, $\gamma_{j}(C_{j}) \geq \gamma_{j}(C_{k}) \geq 0$.
Note that subsequently $\lambda_{j}$ may be decreased; however, we will argue that $\gamma_{j}(C_{j}) \geq 0$.
Job $j$ may be added in $A$ in either line~2 or in line~6(i) of the procedure,
thus $\lambda_{j}$ may be decreased only if there is a job $j' > j$ such that $\lambda_{j'} < \gamma_{j}(C_{j'-1})$.
Since $C_{j'-1} \geq C_{j}$ and $\gamma_{j}(t)$ is a decreasing in $t$ we have that $\lambda_{j'}< \gamma_{j}(C_{j})$.
Moreover, $\lambda_{j}$ will continue decreasing until $\lambda_{j'} = \gamma_{j}(C_{j'-1})$ (line~6(ii)).
Note that $\lambda_{j}$ may be decreased in many iterations and thus the condition $\lambda_{j'} = \gamma_{j}(C_{j'-1})$ may hold for more than one job $j'$;
let $j^*$ denote the job of the last iteration for which $\lambda_{j^*} = \gamma_{j}(C_{j^*-1})$.
Since $\lambda_{j^*} \geq 0$ we obtain $\gamma_{j}(C_{j}) \geq \lambda_{j^*} \geq 0$.
Hence we showed that at the end of Procedure~\ref{algo:energy} the dual solution satisfies property $(\mathcal{Q}1)$.

Moreover, property $(\mathcal{Q}2)$ is a relaxed variant of $(\mathcal{P}2)$, therefore Procedure~\ref{algo:linear2} guarantees $(\mathcal{Q}2)$.
In addition, Procedures~\ref{algo:energy} and~\ref{algo:linear2} satisfy $(\mathcal{Q}3)$,
using very similar arguments as in the proof of Lemma~\ref{lemma:completionlinear}.
\end{proof}

The above lemma in conjunction with Theorem~\ref{thm:framework2} lead to the following result.

\begin{theorem} \label{thm:concave}
HDF is a $(1+\epsilon)$-speed $\frac{1+\epsilon}{\epsilon}$-competitive for integral GFP with concave cost functions.
\end{theorem}

\section{Online JDGFP with differentiable concave cost functions}
\label{sec:jdgfp}

We consider the online JDGFP, assuming that for each job $j$, the
cost function $g_{j}$ is concave and differentiable. Instead of analyzing the fractional objectives and
rounding to integral ones (as in previous sections), we study directly the integral objective by considering
a non-convex relaxation.
Let $x_{j}(t)$ be the variable indicating the execution rate of job $j$ at time $t$.
Let $C_{j}$ be the variable representing the completion time of job $j$. We have the
following non-convex relaxation:
\begin{align*}
\min \sum_{j \in \mathcal{J}} \delta_{j} \int_{r_j}^{C_{j}} g_{j}(C_{j}-r_j) & x_j(t) dt \nonumber\\
\int_{r_j}^{C_{j}} x_j(t) dt &= p_j & & \forall j \in \mathcal{J} \\
\sum_{j \in \mathcal{J}} x_j(t) &\leq 1 & & \forall t \geq 0 \\
x_j(t) &\geq 0 & & \forall j \in \mathcal{J}, t \geq 0
\end{align*}
By associating the dual variables $\lambda_j$ and $\gamma(t)$ with the first and second constraints respectively,
we obtain the Lagrangian dual program \linebreak $\max_{\lambda, \gamma} \min_{x,C} L(\lambda,\gamma,x,C)$
where the Lagrangian function $L(\lambda,\gamma,x,C)$ is equal to\\
\begin{align*}
L(\lambda,\gamma,x,C)=  \sum_{j} \delta_{j} \int_{r_j}^{C_{j}} g_{j}(C_{j}-r_j) x_j(t) dt
 &+ \sum_{j} \lambda_{j} \biggl( p_{j} - \int_{r_j}^{C_{j}} x_j(t) dt \biggl) \\
 &+ \int_{0}^{\infty} \gamma(t) \biggl( \sum_{j} x_j(t) - 1\biggl) dt.
\end{align*}
Next, we describe the algorithm and its analysis, which is inspired by~\cite{ImKulkarni14:Competitive-Algorithms,ImKulkarni14:SELFISHMIGRATE:-A-Scalable}.


\medskip
\noindent
{\bf Algorithm.} \ Let $k = \lceil 2/\epsilon \rceil$.
We write $j' \prec j$ if $r_{j'} \leq r_{j}$, breaking ties arbitrarily.
Let $G_{j}(t) = \sum_{a \preceq j} w_{a}g'_{a}(t-r_{a})$ where the sum is taken over pending jobs
$a$ at time $t$. Note that as $g_{a}$'s are concave functions,
$G_{j}(t)$ is non-increasing (with respect to $t$).
Informally, each term $w_{a}g'_{a}(t-r_{a})$ represents the rate of the contribution
of pending job $a$ at time $t$ to the total cost.
Thus $G_{j}(t)$ stands for the contribution rate of pending
jobs released before $j$ (including $j$) at time $t$ to the total cost.
Whenever it is clear from context, we omit the time parameter.
At time $t$, the algorithm processes job $j$ at rate proportional to $G_{j}(t)^{k} - G_{j+1}(t)^{k}$.
In other words, the rate of job $j$ at time $t$ is
$
\nu_{j}(t) = (G_{j}(t)^{k} - G_{j+1}(t)^{k})/G(t)^{k},
$
where $G(t) = \sum_{a} w_{a}g'_{a}(t-r_{a})$, and the sum is taken over pending jobs $a$ at time $t$.

Next, we define the dual variables. Define $\gamma(t) = 0$.
Moreover, define $\lambda_{j}$ such that
$
\lambda_{j} p_{j} = \frac{1}{k+1} \int_{r_{j}}^{C_{j}} \biggl( \nu_{j}(t) \sum_{a \preceq j} w_{a} g'_{a}(t-r_{a})
                                + w_{j}g'_{j}(t-r_{j}) \sum_{a \prec j} \nu_{a}(t) \biggl) dt,
$
where $C_{j}$ is the completion time of job $j$.
In the following we will bound the total cost of the algorithm by the Lagrangian dual value.
Let $\mathcal{F}$ denote the total cost of the algorithm. We will use the definitions of the
dual variables to derive first some essential properties (Lemma~\ref{lemma:jd.1} and~\ref{lem:job-dep-cost-cru}).

\begin{lemma}
$(k+1)\sum_{j}\lambda_{j}p_{j} = \int_{0}^{\infty} G(t)dt = \mathcal{F}$.
\label{lemma:jd.1}
\end{lemma}
\begin{proof}
We first show that $(k+1)\sum_{j}\lambda_{j}p_{j} = \mathcal{F}$. Consider the term
$w_{j} g'_{j}(t-r_{j})$ in the sum $(k+1)\sum_{j}\lambda_{j}p_{j}$ for $r_{j} \leq t \leq C_{j}$.
The coefficient of this term in the sum is equal to $\sum_{a \prec j} \nu_{a}(t) + \sum_{a \succeq j} \nu_{a}(t) = 1$.
Therefore,
$$
(k+1)\sum_{j}\lambda_{j}p_{j} = \sum_{j} \int_{r_{j}}^{C_{j}} w_{j} g'_{j}(t-r_{j}) dt
    = \sum_{j} w_{j} g_{j}(C_{j}-r_{j}).
$$

Last, note that the identity $\int_{0}^{\infty} G(t)dt = \sum_{j} w_{j} g_{j}(C_{j}-r_{j})$ is straightforward from the definition of $G(t)$.
\end{proof}

\begin{lemma} \label{lem:job-dep-cost-cru}
For any time $\tau \geq r_{j}$, it holds that
$\lambda_{j} \leq \delta_{j} g_{j}(\tau-r_{j}) + \frac{1}{k} G(\tau)$.
\end{lemma}
\begin{proof}
The proof is by the same scheme as in~\cite{ImKulkarni14:SELFISHMIGRATE:-A-Scalable}.
We nevertheless present a complete proof, which
follows from the following two claims, namely Claim~\ref{claim:job-dep-cost-1} and Claim~\ref{claim:job-dep-cost-2}.

\begin{claim} \label{claim:job-dep-cost-1}
For any time $\tau \geq r_{j}$,
$$
\frac{1}{p_{j}} \int_{r_{j}}^{\tau} \biggl( \nu_{j}(t) \sum_{a \preceq j} w_{a} g'_{a}(t-r_{a})
                        + w_{j}g'_{j}(t-r_{j}) \sum_{a \prec j} \nu_{a}(t) \biggl) dt
                \leq (k+1) \delta_{j} g_{j}(\tau-r_{j})
$$
\end{claim}
\begin{claimproof}
At time $t$,
\begin{align*}
\nu_{j}(t) \sum_{a \preceq j} w_{a} g'_{a}(t-r_{a})
    = \frac{G_{j}(t)^{k} - G_{j+1}(t)^{k}}{G(t)^{k}} G_{j}(t)
    \leq k w_{j} g'_{j}(t-r_{j})
\end{align*}
where the inequality is due to the convexity of function $z^{k}$.
Moreover, note that $\sum_{a \prec j} \nu_{a} \leq 1$. Therefore,
\begin{align*}
\int_{r_{j}}^{\tau} \biggl( \nu_{j}(t) &\sum_{a \preceq j} w_{a} g'_{a}(t-r_{a}) + w_{j}g'_{j}(t-r_{j}) \sum_{a \prec j} \nu_{a}(t) \biggl) dt \\
    &\leq \int_{r_{j}}^{\tau} (k+1) w_{j} g'_{j}(t-r_{j}) dt 
    = (k+1) w_{j} g_{j}(\tau-r_{j})
\end{align*}
The claim follows.
\end{claimproof}

\begin{claim} \label{claim:job-dep-cost-2}
For any time $\tau \geq r_{j}$,
$$
\frac{1}{p_{j}} \int_{\tau}^{C_{j}} \biggl( \nu_{j}(t) \sum_{a \preceq j} w_{a} g'_{a}(t-r_{a})
                + w_{j}g'_{j}(t-r_{j}) \sum_{a \prec j} \nu_{a}(t) \biggl) dt
                \leq \biggl(1+ \frac{1}{k} \biggl) G(t)
$$
\end{claim}
\begin{claimproof}
We have
\begin{align*}
&\frac{1}{p_{j}}\int_{\tau}^{C_{j}} \biggl( \nu_{j}(t) \sum_{a \preceq j} w_{a} g'_{a}(t-r_{a})
            + w_{j}g'_{j}(t-r_{j}) \sum_{a \prec j} \nu_{a}(t) \biggl) dt \\
&= \frac{1}{p_{j}}\int_{\tau}^{C_{j}} \nu_{j}(t) \biggl( \sum_{a \preceq j} w_{a} g'_{a}(t-r_{a})
    + w_{j}g'_{j}(t-r_{j})  \frac{G_{j+1}(t)^{k}}{G_{j}(t)^{k} - G_{j+1}(t)^{k}} \biggl) dt \\
&\leq \frac{1}{p_{j}}\int_{\tau}^{C_{j}} \nu_{j}(t) \biggl( \sum_{a \preceq j} w_{a} g'_{a}(t-r_{a})
    + \frac{1}{k} G_{j+1}(t) \biggl) dt \\
&\leq \frac{1}{p_{j}}\biggl(1+ \frac{1}{k} \biggl) G_{j}(\tau) \int_{\tau}^{C_{j}} \nu_{j}(t) dt
= \frac{1}{p_{j}}\biggl(1+ \frac{1}{k} \biggl) G_{j}(\tau) p_{j} \\
&= \biggl(1+ \frac{1}{k} \biggl) G_{j}(\tau) \leq \biggl(1+ \frac{1}{k} \biggl) G(\tau)
\end{align*}
where the first inequality is due to the convexity of $z^{k}$ and the second inequality follows from the
fact that $G_{j}(t)$ is non-increasing in $t$.
Hence, the claim follows.
\end{claimproof}

The proof of the lemma follows by the above two claims.
\end{proof}

\begin{theorem}\label{thm:lagrangian}
The algorithm is $(1+\epsilon)$-speed $\frac{4(1+\epsilon)^2}{\epsilon^{2}}$-competitive for integral JDGFP.
\end{theorem}
\begin{proof}
With our choice of dual variables, the Lagrangian dual objective is
\begin{align*}
\min_{x,C} & \sum_{j} \lambda_{j} p_{j} - \int_{0}^{\infty} \gamma(t)dt
    - \sum_{j} \int_{r_{j}}^{C_{j}} x_{j}(t) \biggl( \lambda_{j} - \gamma(t) - \delta_{j} g_{j}(C_{j}-r_{j}) \biggl) dt \\
&\geq \min_{x,C} \sum_{j} \lambda_{j} p_{j}
    - \sum_{j} \int_{r_{j}}^{C_{j}} x_{j}(t) \biggl( \lambda_{j} - \delta_{j} g_{j}(t-r_{j}) \biggl) dt \\
&\geq \min_{x,C} \sum_{j} \lambda_{j} p_{j}
    - \int_{0}^{\infty} \sum_{j} x_{j}(t) \frac{G(t)}{k} dt
\end{align*}
The first inequality is due to $\gamma(t) = 0$ for every $t$
and $t \leq C_{j}$ in the integral term corresponding to $j$
so $g_{j}(t - r_{j}) \leq g_{j}(C_{j} - r_{j})$ for every job $j$.
The second inequality holds by Lemma \ref{lem:job-dep-cost-cru}.

In the resource augmentation model, the offline optimum has a machine of speed $\frac{1}{1+\epsilon}$, i.e., $\sum_{j} x_{j}(t) \leq \frac{1}{1+\epsilon}$
(while the algorithm has unit-speed machine).
Recall that $k = \lceil 2/\epsilon \rceil$.
Therefore, the Lagrangian dual is at least
\begin{align*}
\sum_{j} \lambda_{j}p_{j} - \frac{1}{1+\epsilon}\int_{0}^{\infty} \frac{G(t)}{k} dt
    = \frac{1}{k+1} \mathcal{F} - \frac{1}{1+\epsilon} \frac{1}{k} \mathcal{F}
    \geq \frac{\epsilon^{2}}{4(1+\epsilon)^2} \mathcal{F}
\end{align*}
The algorithm has cost $\mathcal{F}$ so the theorem follows.
\end{proof}

\section{Conclusion}
\label{sec:conclusion}
In this work we applied primal-dual and dual-fitting techniques in the analysis of online algorithms
for generalized flow-time scheduling problems. This approach yields proofs that are derived from duality principles,
unlike previous approaches that are predominantly based on potential functions. More importantly, we showed how
to exploit duality in order to bypass a canonical rounding of fractional solutions that has been, up to now,
a standard tool in the area of online scheduling. As a result, we obtained, at least for some objective functions,
improved competitive ratios.

A promising direction for future work is to apply our framework to {\em non-clairvoyant} problems.
It would be very interesting to obtain a primal-dual analysis of
Shortest Elapsed Time First (SETF)
which is is known to be scalable~\cite{KalyanasundaramPruhs00:Speed-is-as-powerful}; moreover,
this algorithm has been analyzed in~\cite{FoxIm13:Online-Non-clairvoyant} in the context of the
online GFP with convex/concave cost functions. We believe that one can use duality to argue
that SETF is the non-clairvoyant counterpart of HDF; more precisely, we believe that one can derive SETF as a primal-dual
algorithm in a similar manner as the discussion of HDF in Section~\ref{section:flow}. A more challenging task is to bound
the primal and dual objectives, which appears to be substantially harder than in the clairvoyant setting.
A further open question is extending the results of this paper to multiple machines; here, one potentially
needs to define the dual variable $\gamma(t)$ with respect to as many curves per job as machines.
Last, it would be very interesting to extend the framework in order to allow for algorithms that
are not necessarily scalable. To this direction, one needs to further relax the
conditions of the proposed framework, so as to exploit the speed augmentation and remedy the problematic situations
in which primal job contributions correspond to a negative contribution in the dual.

\section*{Acknowledgements}
Spyros Angelopoulos is supported by project ANR-11-BS02-0015 ``New Techniques in Online Computation--NeTOC''.
Giorgio Lucarelli is supported by the ANR project Moebus (Grant No.
ANR-13-INFR-0001).
Nguyen Kim Thang supported by the FMJH program Gaspard Monge in optimization and operations research and by EDF.

\newpage

\appendix
{\LARGE \bf Appendix}

\section{LP-formulation of fractional JDGFP}
\label{app:lp.formulation}

We argue that the following LP is a relaxation of JDGFP.

\begin{align*}
\min \sum_{j \in \mathcal{J}} \delta_j \int_{r_j}^{\infty} g_j(t-r_j) & x_j(t) dt \\
\int_{r_j}^{\infty} x_j(t) dt &\geq p_j & & \forall j \in \mathcal{J} \\
\sum_{j \in \mathcal{J}} x_j(t) &\leq 1 & & \forall t \geq 0 \\
x_j(t) &\geq 0 & & \forall j \in \mathcal{J}, t \geq 0 \notag
\end{align*}

Consider a job $j \in \mathcal{J}$.
By definition, the fractional weighted cost of a job $j \in \mathcal{J}$ is equal to
\begin{eqnarray*}
\int_{r_j}^{\infty} w_j(t) g_j'(t-r_j) dt
 & = & \frac{w_j}{p_j} \int_{r_j}^{\infty} q_j(t) g_j'(t-r_j) dt\\
 & = & \frac{w_j}{p_j} \int_{r_j}^{\infty} \left( g_j'(t-r_j) \int_{u=t}^{\infty} x_j(u) du \right) dt
\end{eqnarray*}
By changing the order of the integrals we get
\begin{eqnarray*}
\int_{r_j}^{\infty} w_j(t) g_j'(t-r_j) dt
 & = & \frac{w_j}{p_j} \int_{u=r_j}^{\infty} \left( x_j(u) \int_{t=r_j}^{u} g_j'(t-r_j) dt\right) du\\
 & = & \frac{w_j}{p_j} \int_{u=r_j}^{\infty} g_j(u-r_j) x_j(u) du
\end{eqnarray*}
\qed

\end{document}